%% file: lpsize.tex
\documentclass[11pt,letterpaper]{article}

\def\showauthornotes{0}
\def\showtableofcontents{0}
\def\showkeys{0}
\def\showdraftbox{0}
\def\showcolorlinks{1}
\def\usemicrotype{1}
\def\showfixme{0}
\def\focs{0}

\input{macros}

\let\pref=\prettyref

\newcommand{\inst}{\Im}

\renewcommand{\Ind}{\mathds 1}

\let\Rnn\varvarRnn

\newcommand{\wt}{\mathrm{wt}}
\newcommand{\symdif}{\oplus}
\newcommand{\orb}{\mathrm{Orb}}

\DeclareMathOperator*{\pE}{\tilde \E}

\title{\bfseries Approximate constraint satisfaction \\ %
  requires large LP relaxations}

\author{Siu On Chan\footnote{Microsoft Research New England}   \and \hspace*{-0.15in} James R. Lee\footnote{University of Washington}  \hspace*{-0.15in}\and Prasad Raghavendra\footnote{U. C. Berkeley} \and \hspace*{-0.15in} David Steurer\footnote{Cornell University}}

\date{}

\begin{document}

\maketitle
\thispagestyle{empty}

\begin{abstract}
  We prove super-polynomial lower bounds on the size of linear
  programming relaxations for approximation versions of constraint
  satisfaction problems.
We show that for
these problems,
polynomial-sized linear programs are
no more powerful than 
 programs arising from a constant number of
rounds of the Sherali--Adams hierarchy.

In particular,  any polynomial-sized
linear program for \maxcut has an integrality gap of $\frac12$
and any such linear program for \maxthreesat has an integrality gap of $\frac78$.
\end{abstract}



\ifnum\showtableofcontents=1
{
\tableofcontents
\thispagestyle{empty}
 }
\fi


\renewcommand{\P}{\mathrm P}
\newcommand{\NP}{\mathrm{NP}}

\section{Introduction}
\label{sec:introduction}

Linear programming is one of the most powerful tools known
for finding approximately optimal solutions to $\NP$-hard problems.
We refer to the books \cite{VaziraniBook01,WilliamsonShmoysBook11}
which each contain a wealth of examples.
If $\P \neq \NP$, then for many such problems
we do not expect polynomial-sized linear programs (LPs) to compute
arbitrarily good approximations to the optimal solution.
(More formally, if $\NP \nsubseteq \P/\mathrm{poly}$, then such LPs
cannot exist \cite{Yannakakis91}.)

Thus a line of research has sought to prove lower bounds
on the efficacy of small linear programs.
The construction of integrality gaps for specific LPs
has long been a topic of interest in approximation algorithms.
Arora, Bollob\'as, and Lov\'asz \cite{AroraBollobasLovasz02} initiated
a more systematic study; they explored the limitations of LPs arising
from lift-and-project hierarchies like those of Lov\'asz and Schrijver \cite{LovaszSchrijver1991}
and Sherali and Adams \cite{SheraliAdams1990}.
There has now been an extensive amount of progress
made in this area; one can see a sampling in the section
on previous work.

Arguably, the ultimate goal of this study is to prove unconditional lower
bounds for {\em every} sufficiently small LP.  Since linear programming
is $P$-complete under various notions of reduction, this
would require proving that $\NP$ does not have polynomial-size circuits (see, e.g.,
the discussion in \cite{Yannakakis91}).
But one could still hope to complete this program
for LPs that use the natural encoding
of the underlying combinatorial problem.

We make progress toward this goal for the class of constraint satisfaction problems (CSPs).
For instance, we prove that every polynomial-sized LP for \maxcut has an integrality gap of $\frac12$,
answering a question from \cite{BraunFPS12}.
As another example, every such LP   for \maxthreesat has an integrality gap of $\frac78$,
and every such LP for \maxtwosat has an integrality gap of $\frac34$.
In fact, in both cases these integrality gaps hold for families of LPs of size up to $\cramped{n^{o(\frac{\log n}{\log \log n})}}$.

Corresponding
upper bounds for all three problems can be achieved by simple polynomial-sized LPs.  For \maxthreesat,
a $\frac78$-approximation is best-possible assuming $\P \neq \NP$ \cite{Hastad01}.  For \maxcut,
the seminal SDP-based algorithm of Goemans and Williamson \cite{GoemansWilliamson95}
achieves a $0.878$-approximation.  In this case, our result yields
a strict separation between the power of polynomial-sized LPs and SDPs
for a natural optimization problem.
Interestingly, even a simple spectral algorithm can do strictly better than $1/2$ for \maxcut \cite{Trevisan2012}.
\Dnote{In this sense, our results shows that LP relaxations fail capture arguments based on eigenvalues in a strong way.
OR: our results give a strong separation between combinatorial  / LP algorithms and spectral / SDP algorithms.}

To establish these lower bounds, we show that for approximating CSPs,
polynomial-sized LPs are exactly as powerful as those
programs arising from $O(1)$ rounds of the Sherali--Adams hierarchy.
We are then able to employ the powerful Sherali--Adams gaps that
appear in prior work.  This offers a potential framework
for understanding the power of linear programs
for many problems by relating their expressive power
to that of the very explicit Sherali--Adams hierarchy.
\Dnote{maybe something stronger:
our results show that for the purpose of approximating constraint satisfaction problems Sherali--Adams LP algorithms are universal among all possible LP algorithms of comparable complexity.}

\medskip

In \pref{sec:maxcut}, we discuss our approach
for the specific example of \maxcut, including
the class of LPs to which our results apply.
\pref{sec:preliminaries} is devoted
to a review of CSPs and their linear relaxations.
There we explain our basic approach to proving lower bounds
by exhibiting an appropriate separating hyperplane.  We also review
the Sherali--Adams hierarchy for CSPs.
In \pref{sec:LPs},  we present the technical components
of our approach, as well as the proof of our main theorem.

Finally, \pref{sec:symm-formulations} contains an illustrative
discussion of how Sherali--Adams gap examples
can be used to construct corresponding gaps for {\em symmetric} LPs.
This connection is quantitatively stronger than our result for general LPs.
We refer to \pref{sec:conclusion} for
a discussion of future directions.

\medskip
\noindent
{\bf Recent work.}
Since initial publication of this manuscript,
there has been substantial followup work building
on the ideas presented here.
The papers \cite{LRST14,FSP14} establish a connection between
{\em symmetric} semidefinite programs and the Sum-of-Squares
hierarchy by analogy with our work in \pref{sec:symm-formulations}.
In \cite{LRS15}, a connection between {\em general}
semidefinite extended formulations and the Sum-of-Squares
hierarchy is established; in particular, the authors
prove exponential lower bounds on the semidefinite
extension complexity of explicit polytopes (like the TSP polytopes).
Finally, our models for approximation via linear programs
are extended and refined in the work \cite{BPZ15};
the authors show that a suitable notion of reduction within
the model allows one to derive lower bounds for
additional problems (other than CSPs).

\subsection{History and context}

\paragraph{Extended formulations}
In a seminal paper, Yannakakis \cite{Yannakakis91} proved that every
symmetric LP (i.e., one whose formulation is invariant under
permutations of the variables) for TSP has exponential size.
Only recently was a similar lower bound given for general LPs.
More precisely, Fiorini, et\ al.\ \cite{FMPTW12} show that
the extension complexity of the TSP polytope is at least $2^{\Omega(\sqrt{n})}$
for $n$-vertex graphs.

Braun, et\ al.\ \cite{BraunFPS12} expand the notion of
extension complexity to include approximation problems
and show that approximating \maxclique within $O(n^{1/2-\e})$
requires LPs of size $2^{\Omega(n^{\e})}$.
Building on that work, Braverman and Moitra \cite{BravermanMoitra13}
show that approximating \maxclique within $O(n^{1-\e})$
requires LPs of size $2^{\Omega(n^{\eps})}$.
We remark that the encoding of \maxclique used
in the later two works is somewhat lacking.
Specifically, these lower bounds do not encompass, for instance,
standard relaxations for \maxclique, including
those given by the Sherali-Adams hierarchy.

These three latter papers all use Yannakakis' connection
between extension complexity and non-negative rank (see \cite{FMPTW12}
for a detailed discussion).  They are based on
increasingly more sophisticated analyses of a single family
of slack matrices first defined in \cite{FMPTW12} (and extended
to the approximation setting by \cite{BraunFPS12}).
Closely related slack matrices are employed in a recent paper of
Rothvoss \cite{Rothvoss14} to show exponential lower bounds
on the extension complexity of the matching polytope.
A significant contribution of the present work is that the
connection between general LPs and the Sherali--Adams hierarchy
allows one to employ a much richer family
of hard instances.

\Dnote{maybe explain at the example of vertex cover / independent set.
what kind of LP relaxations are captured by the framework and what it means to encode the problem in different ways.}

\paragraph{LP and SDP hierarchies}
As mentioned previously, starting with the works \cite{AroraBollobasLovasz02,ABLT2006}, the
efficacy of LP and SDP hierarchies for approximation problems has been
extensively studied.  We refer to the survey of Laurent \cite{laurent2003comparison}
for a discussion of the various hierarchies and their relationships.

We mention a few results that will be quite useful for us.
Fern\'andez de la Vega and Mathieu \cite{FernandezKenyonMathieu07}
showed that for any fixed $\e > 0$ and $k$,
\maxcut has an integrality gap of $\frac12+\e$ even
after $k$ rounds of the Sherali--Adams hierarchy.
In a paper of Charikar, Makarychev, and Makarychev \cite{charikar2009integrality},
it is shown that \maxcut and \vertexcover have integrality
gaps of $\frac12+\e$ and $2-\e$, respectively, for $n^{\Omega(\e)}$ rounds
of the Sherali--Adams hierarchy.

In work of Schoenebeck \cite{schoenebeck2008linear}, tight bounds are given on
the number of rounds needed to approximate $k$-CSPs in the Lasserre
hierarchy (which, in particular, is stronger than the Sherali--Adams hierarchy).
For instance, he shows that for every $\e > 0$, \maxthreesat has a
$\frac{7}{8}+\e$ integrality gap even after $\Omega(n)$ rounds.
One should consult also the much earlier work of Grigoriev \cite{Grigoriev2001}
which achieves an equivalent family of lower bounds stated in the dual
setting of Positivstellensatz proof systems.
There are also Sherali--Adams integrality gaps for CSPs with a pairwise
independent predicate, due to Benabbas et\ al.\ \cite{BenabbasGMT12}.

\paragraph{Strong separation between nonnegative rank and smooth nonnegative rank}
We remark that all previous lower bounds for nonnegative rank (at least in the context of extended formulations) are robust with respect to small multiplicative perturbations \cite{Rothvoss14,FMPTW12,BraunFPS12,BravermanMoitra13}.
Concretely, if we define the \(\e\)-smooth nonnegative rank of a matrix \(A\) as \[
{\mathrm{rank}_{+,\e}}(A) \seteq \min  \left\{ {\mathrm{rank}_{+}}(A') \mid (1-\e) A_{ij} \le A'_{ij} \le (1+\e)A_{ij} \right\}
\mcom
\] then all previous lower bounds for nonnegative rank also lower bound the \(\e\)-smooth version for some absolute constant \(\e>0\).
A related generalization of nonnegative rank is approximate nonnegative rank that allows additive instead of multiplicative error.
\footnote{All previous lower bounds also hold for this generalization of nonnegative rank.
However, some of the lower bound arguments do not apply to the additive-error setting.
(For example, arguments that rely on the zero / non-zero pattern of the matrix.)}
This version of nonnegative rank is equivalent to the smooth rectangle bound \cite{kol2014approximate}.

In contrast, the matrices studied in this work turn out to have only polynomial approximate and smooth nonnegative rank.
In this sense, our superpolynomial lower bounds on the nonnegative rank of these matrices give the first separation between nonnegative rank and smooth nonnegative rank.
See \pref{sec:smoothrank} for a discussion.


%

\subsection{Outline:  \maxcut}
\label{sec:maxcut}
\newcommand{\sa}{\mathtt{SA}}

We now present the basic details of our approach applied to the \maxcut problem.
To this end, consider a graph $G=(V,E)$ with $|V|=n$.  For any $S \subseteq V$, we use
$$
G(S) \defeq \frac{|E(S,\bar S)|}{|E|}
$$
to denote the fraction of edges of $G$ crossing the cut $(S,\bar S)$.
The maximum cut value of $G$ is $\opt(G) = \max_{S \subseteq V} G(S)$.

\medskip
\noindent
{\bf The standard LP.}
To construct an LP for computing (or approximating) $\opt(G)$,
it is natural to introduce variables $x = (x_1, x_2, \ldots, x_n) \in \sbits^n$
corresponding to the vertices of $G$.
One can then write, for instance,
$$
\opt(G) = \max_{x \in \sbits^n} \frac{1}{|E|} \sum_{\{i,j\} \in E} \frac{1-x_i x_j}{2}\,.
$$
To convert this computation into a linear program, we need to {\em linearize} the
objective function.

The usual way is to introduce new LP variables $y = (y_{i,j}) \in \mathbb R^{n \choose 2}$ meant to represent
the quantities $(1-x_ix_j)/2$.
Now consider the vector ${v}_G \in \{0,1\}^{n \choose 2}$ such that $(v_G)_{\{i,j\}} = 1$ precisely
when $\{i,j\} \in E$.  Given that we have linearized both the graph $G$ and the cut variable $x$,
we can consider the LP relaxation
$$
\mathcal L(G) = \max_{y \in P} \,\langle v_G, y \rangle\,,
$$
where $P$ is any polytope containing all the vectors $y$ such that $y_{i,j} = (1-x_i x_j/2)$ for some $x \in \sbits^n$.
The standard relaxation corresponds to a polytope $P$ defined by the constraints $\{ 0 \leq y_{i,j} \leq 1 : i,j \in V \}$ and
\[\left\{ y_{i,j} \leq y_{i,k} + y_{k,j},\,\, y_{i,j} + y_{i,k} + y_{k,j} \leq 2  : i,j,k \in V \right\}\,.\]
Clearly $P$ is characterized by $O(n^3)$ inequalities.

\medskip
\noindent
{\bf Arbitrary linearizations.}  But it is important to point out that, for our purposes,
any linearization of the natural formulation of \maxcut suffices.
We only require that
there is a number $D \in \mathbb N$ such that:
\begin{enumerate}
\item For every graph $G$, we have a vector $v_G \in \mathbb R^D$.
\item For every cut $S \subseteq V$, we have a vector $y_S \in \mathbb R^D$.
\item For all graphs $G$ and vectors $y_S$, the condition $G(S) = \langle v_G, y_S \rangle$ holds.
\end{enumerate}
Now any polytope $P \subseteq \mathbb R^D$, such that $y_S \in P$ for every $S \subseteq V$,
yields a viable LP relaxation:  $\mathcal L(G) = \max_{y \in P} \langle v_G, y\rangle$.
The {\em size} of this relaxation is simply the number of facets of $P$,
i.e. the number of linear inequalities needed to specify $P$.
\Dnote{this discussion is unclear.
we are talking about two different things.
the ``formulation'' of a problem and the ``linearization'' of the formulation.
the term ``formulation'' is not defined but important.
when we say that all linearizations are equivalent it's a bit subtle.
if we pick a fixed linearization then we might need to allow for projections.}

\begin{remark}
We stress that the polytope $P$ depends only on the input size.  This
is akin to lower bounds in non-uniform models of computation like
circuits wherein there is a single circuit for all inputs of a certain
size.  The input graph $G$ is used only to define the objective
function being maximized.  In other words, the variables and constraints of the
linear program are fixed for each input size while the objective
function is defined by the input.  To the best of our knowledge, all
linear and semidefinite programs designed for approximating max-CSP problems are subsumed by
relaxations of this nature.
\end{remark}

In \pref{sec:LPs}, we prove that every such relaxation of polynomial size
has an integrality gap of $\frac12$ for \maxcut.
We now give an informal outline of the proof.

\medskip
\noindent
{\bf Proving a lower bound.}
In \pref{thm:lp-characterization},
we recall that if there is an LP relaxation $\mathcal L$ of size $R$, then
a simple application of Farkas' Lemma shows that
there are non-negative functions $q_1, \ldots, q_R : \sbits^n~\to~\mathbb R_{\geq 0}$
such that for every graph $G$,
there are coefficients $\lambda_0, \lambda_1, \ldots, \lambda_R \geq 0$ satisfying
\begin{equation}\label{eq:combination}
\mathcal L(G) - G(x) = \lambda_0 + \lambda_1 q_1(x) + \cdots + \lambda_R q_R(x)\,.
\end{equation}
for all $x \in \sbits^n$.
(Note that we have earlier viewed $G$ as a function on cuts and we now view
it as a function on $\sbits^n$ by associating these vectors with cuts.)

One should think of \eqref{eq:combination} as saying that
$\mathcal L(G)-G \in \cone(\mathds{1}, q_1, q_2, \ldots, q_R)$, where the latter object
is the cone generated by $\{\mathds{1},q_1, q_2, \ldots, q_R\}$
inside the Hilbert space $L^2(\sbits^n)$ of real-valued functions, and $\mathds{1}$ denotes the function
that is identically 1.
These functions $q_i : \sbits^n \to \mathbb R_{\geq 0}$ encode the slack of
each constraint of the LP.  Thus if the $i$th LP constraint is of the form
$\langle A_i ,z\rangle \leq b_i$, then $q_i(x) = b_i - \langle A_i, y_{S_x}\rangle$ where
$y_{S_x}$ is the cut vector corresponding to $x \in \sbits^n$.
\medskip

Consider some $m \ll n$.
The $d$-round Sherali--Adams relaxation for an $m$-vertex graph $G_0$ has value
$\sa_d(G_0) \leq c$ if and only if there exist a family of non-negative $d$-juntas $\{f_i~:~\sbits^m \to \mathbb R_{\geq 0}\}$
such that
\begin{equation}\label{eq:SAexp}
c-G_0 = \sum_i \lambda_i f_i\,,
\end{equation}
where $\lambda_i \geq 0$ for each $i$.
We recall that a $d$-junta is a function whose value
depends on at most $d$ of its inputs.
See \pref{sec:Sherali--Adams-lp} for an explanation of \eqref{eq:SAexp}.

\Jnote{Better explanation of this needed e.g. Toni's question}

In particular, if $G_0$ is such that $\sa_d(G_0) > c$, then
no such representation \eqref{eq:SAexp} with $d$-juntas can exist.  Our goal is to use
\eqref{eq:combination}
to find a graph $G$ on $n$ vertices such that $\opt(G)=\opt(G_0)$, and
such that $G_0$ has a representation of the form \eqref{eq:SAexp} with $c=\mathcal L(G)$.
This will show that $\mathcal L(G) \geq \sa_d(G_0)$, completing our proof.
(Recall that since we are dealing with maximization problems and $\opt(G)=\opt(G_0)$,
this means that our LP is not doing better than Sherali--Adams.)

This proceeds in three steps:  First, we argue that, by a truncation argument, it suffices to consider functions $\{q_i\}$
that are sufficiently smooth.
  Then in \pref{sec:high-entr-distr}, we show that any sufficiently smooth $q_i$
can be approximated (in a certain weak sense) by a $K$-junta $q'_i$ for $K$ which may be quite large (e.g., $K=n^{0.2}$).

In \pref{sec:rand-restr-proof}, we employ a random restriction argument:  By planting the $m$-vertex instance $G_0$ at random
inside a larger graph $G$ (on $n$ vertices), we can ensure that for every $q'_i$, the set of significant coordinates
{\em when restricted to $G_0$} is much smaller; in fact, we show that with high probability
over the random planting, every such $q_i'$ has only $d$ significant
coordinates in the support of $G_0$.
Here we use crucially the fact that we have only $R$ functions $\{q_i\}$,
where $R \leq n^{\alpha d}$ for some small constant $\alpha > 0$.

In particular, applying \eqref{eq:combination} to $G$ and then restricting our attention
to the vertices in $V(G_0)$, this yields a representation of the form
\begin{equation}\label{eq:bound}
\mathcal L(G) - G_0 = \lambda_0 + \sum_{i=1}^R \lambda_i q_i|_{V(G_0)}\,,
\end{equation}
and, when restricted to $G_0$, every $q_i$ is weakly approximated by a $d$-junta $q'_i$.
More specifically, all the low-degree Fourier coefficients of $q_i-q'_i$ are small.
Now, the fact that \eqref{eq:bound} holds and each $q_i$ is approximately a $d$-junta
will yield that $\mathcal L(G) \geq \sa_d(G_0)$, taking \eqref{eq:SAexp} into consideration.
Here we remain vague, but the reader should note that this implication
would follow immediately if each $q_i|_{V(G_0)}$ were actually a $d$-junta.

This will hold true as long as the ``approximation''
does not hurt us too much.
One might think that our approximation is too weak:  We only know that $q'_i$ approximates $q_i$
on $V(G_0)$ in the low-degree part.
Now we use the fact that the $d$-round Sherali--Adams relaxation is only
capable of perceiving low-degree functions (more technically,
the $d$-round Sherali--Adams functional introduced in
\pref{sec:Sherali--Adams-lp} is a degree-$d$ multilinear polynomial).  In particular, it suffices that the low-degree parts
of $q_i$ and $q'_i$ are close.

The
ingredients are all put together in \pref{sec:putt-things-togeth},
where one can find the proof of our main theorem for general CSPs.

\section{Background} 
\label{sec:preliminaries}
\label{sec:background}

We now review the maximization versions of boolean CSPs,
their linear programming relaxations, and related issues.

Throughout the paper, for a function $f : \sbits^n \to \mathbb R$,
we write $\E f = 2^{-n} \sum_{x \in \sbits^n} f(x)$.
If $g : \sbits^n \to \mathbb R$, we denote
the inner product $\langle f,g\rangle = \E [fg]$
on the Hilbert space $L^2(\{-1,1\}^n)$.
Recall that any $f : \sbits^n \to \mathbb R$ can be written uniquely
in the Fourier basis as
$f = \sum_{\alpha \subseteq [n]} \hat f(\alpha) \chi_{\alpha}$,
where $\chi_{\alpha}(x) = \prod_{i \in \alpha} x_i$ and $\hat f(\alpha) = \langle f,\chi_{\alpha}\rangle$.
A function $f$ is called a {\em $d$-junta} for $d \in [n]$ if $f$ depends only
on a subset $S \subseteq [n]$ of coordinates with $|S| \leq d$.
In other words, $f$ can be written as $f = \sum_{\alpha \subseteq S} \hat f(\alpha) \chi_{\alpha}$.

We say that $f$ is a {\em density} if it is non-negative
and satisfies $\E f = 1$.
For such an $f$, we let $\mu_f$ denote the corresponding
probability measure on $\sbits^n$.
Observe that for any $g : \sbits^n \to \mathbb R$,
we have $\E_{x \sim \mu_f} [g(x)] = \langle f,g\rangle.$

\paragraph{Constraint Satisfaction Problems}

\newcommand{\maxpi}{\problemmacro{Max-$\Pi$}}
\newcommand{\maxthreecnf}{\problemmacro{Max-$3$CNF}}

Constraint satisfaction problems form a broad class of discrete
optimization problems that include, for example, \maxcut and \maxthreesat.
%
For simplicity of presentation, we will focus on constraint satisfaction problems with a boolean alphabet, though
similar ideas extend to larger domains (of constant size).  One can consult
\cite[\S 7]{LRS15}.

For a finite collection $\Pi=\set{P}$ of $k$-ary predicates $P\from
\sbits^k\to \bits$, we let \maxpi denote the following optimization
problem:
An instance $\inst$ consists of boolean variables $X_1,\ldots,X_n$ and
a collection of $\Pi$-predicates $P_1(X),\ldots,P_m(X)$ over
these variables.
A $\Pi$-predicate is a predicate $P_0\from \sbits^n\to \bits$ such
that $P_0(X)=P(X_{i_1},\ldots,X_{i_k})$ for some $P\in\Pi$ and {\em distinct} indices
$i_1,\ldots,i_k\in[n]$.
The objective is to find an assignment $x\in \sbits^n$ that satisfies as
many of the predicates as possible, that is, which maximizes
\begin{displaymath}
  \inst(x)\defeq \tfrac1m\sum_{i=1}^m P_i(x)\mper
\end{displaymath}
We denote the optimal value of an assignment for $\inst$ as
$\opt(\inst)=\max_{x\in\sbits^n}\inst(x)$.

\medskip

\emph{Examples:} \maxcut corresponds to the case where $\Pi$ consists of
the binary inequality predicate.
For \maxthreesat, $\Pi$ contains all eight 3-literal disjunctions,
e.g., $X_1\vee \bar X_2\vee \bar X_3$.

\newcommand{\tinst}{\tilde\inst}
\newcommand{\tx}{\tilde x}

\paragraph{Linear Programming Relaxations for CSPs}
\label{sec:line-progr-relax}

In order to write an LP relaxation for such a problem,
we need to linearize the objective function.
For $n\in \N$, let $\maxpi_n$ be the set of \maxpi instances on $n$ variables.
An {\em LP-relaxation of size $R$} for $\maxpi_n$ consists of
the following.
\begin{description}
\item[Linearization:]
  Let $D$ be a natural number.
  For every $\inst\in\maxpi_n$, we associate a vector $\tinst \in
  \R^D$ and for every assignment $x\in\sbits^n$, we associate a point
  $\tx\in\R^D$, such that $\inst(x)=\iprod{\tinst,\tx}$ for all
  $\inst\in \maxpi_n$ and all $x\in\sbits^n$.
\item[Feasible region:] A closed, convex (possibly unbounded)
	polyhedron $P\sse \R^D$ described by $R$
  linear inequalities, such that $\tx\in P$ for all assignments
  $x\in\sbits^n$.  Note that the polytope $P$ is independent of the instance
  $\inst$ of $\maxpi_n$.
\end{description}

Given an instance $\inst\in \maxpi_n$, the LP relaxation $\mathcal L$
has value \[\mathcal L(\inst)\defeq \max_{y\in P} \iprod{\tinst,y}\,.\]
Since $\tx\in P$ for all assignments $x\in\sbits^n$ and
$\iprod{\tinst,\tx}=\inst(x)$, we have $\mathcal L(\inst)\ge
\opt(\inst)$ for all instances $\inst\in\maxpi_n$.

\begin{remark}
\label{rem:linearization}
For concreteness, there is a ``universal linearization'' for CSPs that one can always use (this is sometimes referred to as the ``vertex extended formulation'').
One views
$x \mapsto \inst(x)$ as a multilinear polynomial over $\sbits^n$.
In the Fourier basis $\{\chi_{\alpha} : \alpha \subseteq [n]\}$, one would
have $\tilde \inst = \sum_{\alpha} \hat \inst(\alpha) \chi_{\alpha}$ and
$\tilde x = \sum_{\alpha} \chi_{\alpha}(x) \chi_{\alpha}$.
Note that if the $\Pi_n$ contains $k$-ary predicates,
then $\tilde \inst$ and $\tilde x$ are multilinear polynomials
of degree at most $k$.
\end{remark}

\begin{remark}
\label{rem:numvars}
Of course, in the preceding linearization, the number of variables is now $2^n$.
But if the number of defining inequalities small, one can reduce
the number of variables via an appropriate linear transformation;
see \cite{FMPTW12}.
\end{remark}

\paragraph{Symmetric Linear Programs}
A symmetric LP is one for which the linearization is
symmetric under any permutation of the input variables.
More precisely, let us suppose $\cL$ is a linear program for
$\maxpi_n$ that associates to each instance $\inst$, a linearization
$\tilde{\inst} \in \R^D$ and to every assignment $x \in \sbits^n$ a
point $\tilde{x} \in \R^D$.

Let $\Sym(n)$ denote the symmetric group on $\{1,2,\ldots,n\}$.
Note that $\Sym(n)$ acts naturally on elements $x \in \mathbb R^n$
by permutation of the coordinates.  Specifically, for a permutation
$\sigma \in \Sym(n)$ and $x \in \sbits^n$, let $\sigma x =
\left(x_{\sigma(1)},x_{\sigma(2)},\ldots, x_{\sigma(n)} \right)$.
This action extends to an
action of $\Sym(n)$ on functions $f : \sbits^n \to \mathbb R$
by defining $\sigma f(x) = f(\sigma x)$ for $\sigma \in \Sym(n)$.

We say that the linear program $\cL$ is
{\it symmetric} if the following holds:
For every permutation $\sigma \in \Sym(n)$,
there exists a corresponding permutation $\tilde{\sigma} \in \Sym(D)$
such that for every assignment $x \in \sbits^n$,
$$\widetilde{\sigma x} =\tilde{\sigma} \tilde{x} \mcom$$
and the feasible region $P \subset \R^D$ remains
invariant under the permutation $\tilde{\sigma}$ of coordinates, i.e.,
$$ \tilde{\sigma} P = P \mper$$
To the best of our knowledge, all linear and semidefinite programming
relaxations designed for approximating max-CSP problems have been
symmetric relaxations.  In general, assymetric relaxatiosn could be
much more powerful as demonstrated by Kaibel et\ al.\ \cite{KaibelPT10}
who show that asymmetric LPs can be superpolynomially smaller than
symmetric LPs for optimizing over $\log{n}$-sized partial matchings.

\medskip
\noindent
{\bf $(c,s)$-approximation.}
For $c > s \geq 0$, we say that a linear programming relaxation $\mathcal L$ for
$\maxpi_n$ is a {\em $(c,s)$-approximation} if $\mathcal L(\inst)\le
c$ for all instances $\inst\in \maxpi_n$ with $\opt(\inst)\le s$.
We also say that $\mathcal L$ achieves an {\em $\alpha$-factor
approximation} if $\mathcal L(\inst)\le \alpha \opt(\inst)$ for all
$\inst\in\maxpi_n$.



The following theorem is inspired by Yannakakis's characterization of
exact linear programming relaxations.
It appears in similar form in previous works
\cite{Pashkovich12} and \cite[Thm. 1]{BraunFPS12}.
For simplicity, we specialize it here for constraint satisfaction problems.
%

\begin{theorem}
\label{thm:lp-characterization}
  For every $c,s \in [0,1]$,
  there exists an LP relaxation of size at most $R$ that
  achieves a $(c,s)$-approximation for $\maxpi_n$ {\em if and only if} there
  exist non-negative
  functions $q_1,\ldots,q_R\from \sbits^n\to\Rnn$ such that for every instance $\inst\in
  \maxpi_n$ with $\opt(\inst)\le s$, the function $c-\inst$ is a
  nonnegative combination of the functions $q_1,\ldots,q_R$ and $\mathds{1}$, i.e.
  \begin{equation}\label{eq:instrep}
      c-\inst \in \Bigset{\lambda_0 + \sum_{i=1}^R \lambda_i q_i\Bigmid \lambda_0, \lambda_1,\ldots,\lambda_R\ge 0}\mper
  \end{equation}
  Moreover, if the LP relaxation is symmetric then there exist
  nonnegative functions $q_1, \ldots, q_R : \sbits^n \to \Rnn$ witnessing \eqref{eq:instrep}
  and such that $\left\{q_1, \ldots, q_R\right\}$ is closed
  under the action of $\Sym(n)$.
 \end{theorem}

\begin{proof}
First, we prove that the existence of an LP relaxation of size $R$ yields
a representation of the form \eqref{eq:instrep}.
Consider a natural number $D$ and linearizations $\tilde\inst, \tilde
x \in \mathbb R^D$
for every $\inst \in \maxpi_n$ and $x \in \sbits^n$.  Let $P \subseteq
\mathbb R^D$
be specified by $R$ linear inequalities $\langle A_i, y\rangle \leq b_i$, and such that
$\tilde x\in P$ for every $x \in \sbits^n$.  We define the function $q_i : \sbits^n \to \mathbb R_+$
by $q_i(x) = b_i - \langle A_i, \tilde x\rangle$.

Consider now any instance $\inst$ with $\opt(\inst) \leq s$.  By assumption, we have $\mathcal L(\inst) \leq c$,
meaning that $c \geq \langle y, \tilde \inst\rangle$ holds for all $y \in P$.
Now Farkas' Lemma \cite[Corollary~5.3c]{AS03} tells us that every valid linear
inequality over $P$ can be written as a non-negative combination of the
inequalities $\{b_i - \langle A_i,y\rangle \geq 0 : i=1,2,\ldots,R\}$, and the
inequality $1 \geq 0$.
This yields the existence of non-negative numbers $\{\lambda_i\}$ such that
$c - \langle y, \tilde \inst \rangle = \lambda_0 + \sum_{i=1}^R \lambda_i (b_i - \langle A_i,y\rangle)$
holds for all $y \in P$.

In particular, this holds for every $\tilde x$, where $x \in \sbits^n$.
Now, a defining property of the linearization is that $\langle \tilde x, \tilde \inst\rangle = \inst(x)$ for
every $x \in \sbits^n$.  Thus we have arrived at a representation of the form \eqref{eq:instrep}.

We now show the reverse implication.
Consider functions $\{q_i\}$ satisfying \eqref{eq:instrep}.
We will let $D=2^n$ and the $D$-dimensional Hilbert space
for our linearization will be $L^2(\{-1,1\}^n)$, which we identify with the linear span of
the Fourier characters $\{\chi_{\alpha} : \alpha \subseteq [n]\}$.
We use the linearization appearing in \pref{rem:linearization}.
We may think of each $q_i$ as lying in $L^2(\{-1,1\}^n)$.
Define a polyhedron $P \subseteq L^2(\{-1,1\}^n)$ by $P = \{ y \in L^2(\{-1,1\}^n) : \langle y,q_i \rangle \geq 0, i=1,\ldots,R\}$.
This yields an LP of size at most $R$ since $\langle \tilde x, q_i\rangle = q_i(x) \geq 0$ for every $i$ and $ x \in \sbits^n$.
Now \eqref{eq:instrep} tells us that whenever $\opt(\inst) \leq s$, the inequality $\langle y, \tilde \inst\rangle \leq c$ is valid
over $P$, implying that $\mathcal L(\inst) \leq c$.  Thus our LP is a $(c,s)$-approximation.

Finally, suppose the LP relaxation is symmetric.  By definition,
for every
$\sigma \in \Sym(n)$, there exists a
$\tilde{\sigma} \in \Sym(D)$ such that
$\widetilde{\sigma x} = \tilde{\sigma} \tilde{x}$
for all $x \in \sbits^n$ and the polytope $P$ is invariant under
the action of
$\tilde{\sigma}$.
We may assume that $P$ is full-dimensional,
and moreover that the facet-defining inequalities
$\iprod{A_i,y}\leq b_i$ are normalized so that $\|A_i\|_2=1$.

Consider an inequality of $\iprod{A_i, y} \leq b_i$ of the
polyhedron $P$ and the corresponding function $q_i : \sbits^n \to \R$
defined by $q_i(x) = b_i - \iprod{A_i, \tilde{x}}$.
Since $P$ is invariant under the action of
$\tilde{\sigma}$, the faces of $P$ are mapped to each other by the
permutation $\tilde{\sigma}$.
Now by our choice of normalization $\|A_i\|_2=1$ for $i=1,2,\ldots, R$,
the
facet-defining inequality
$\iprod{A_i,\tilde{\sigma} y} \leq
b_i$ is the same as $\iprod{A_j, y} \leq b_j$ for some $j \in \{1,2,\ldots,R\}$.
Hence for all $x \in \sbits^n$, $q_i(\sigma  x) = b_i -
\iprod{A_i, \tilde{\sigma} \tilde{x}} = b_j - \iprod{A_j, \tilde x} =
q_j(x)$.  This implies that one can choose the family $Q = \{q_1, \ldots, q_R\}$ of
functions to be invariant under the action of $\Sym(n)$.
\end{proof}

\paragraph{A communication model}
The characterization in \pref{thm:lp-characterization} has an
illustrative interpretation as a two-party, one-way communication complexity
problem:
Alice's input is a \maxpi instance $\inst$ with $\opt(\inst)\le
s$. Bob's input is an assignment $x\in \sbits^n$.
Their goal is to compute the value $\inst(x)$ in expectation.
To this end, Alice sends Bob a randomized message containing at most $L$ bits.
Given the message Bob outputs deterministically a number $v$
such that $v \leq c$.
The protocol is correct if for every input pair $(\inst,x)$, the
expected output satisfies $\E v=\inst(x)$ (the expectation is over
Alice's randomness).

An $L$-bit protocol for this communication problem
yields an LP relaxation of size $2^L$:
If Bob outputs a value $v(x,i)$ based on message $i$ from Alice,
then define $q_i(x)=c-v(x,i)$.  This yields $2^L$ non-negative
functions satisfying the conditions of \pref{thm:lp-characterization}.

On the other hand, if there exist $R=2^L$ functions
$\{q_1, q_2, \ldots, q_R\}$ as in \pref{thm:lp-characterization},
then by adding the constant function $q_{0}$
and an appropriate $\lambda_0 \geq 0$, we may assume
that $\sum_{i=0}^{R} \lambda_i = 1$, i.e. that we have a
convex combination instead of a non-negative combination.
This yields a strategy for Alice and Bob:  Alice
sends an index $i \in \{0,1,\ldots,R\}$, drawn from
a distribution depending on $\inst$ (specified by the coefficients $\{\lambda_i\}$), and then
Bob outputs $c-q_i(x) \leq c$.

\medskip
{\em Example:}
Suppose the optimization problem is \maxcut.
In this case, Alice receives a graph $G=(V,E)$ and Bob a cut $S \subseteq V$.
If Alice sends Bob a uniformly random edge $\{u,v\} \in E$ and Bob outputs
the value $|\Ind_S(u)-\Ind_S(v)|$, the result is a communication (in expectation)
protocol using at most $\log_2 {n \choose 2}$ bits of communication.
In this communication protocol, the value output by Bob is always at
most $1$.  Therefore, this corresponds to a trivial $(1,s)$-approximation for
\maxcut for every $s < 1$.  In any protocol achieving a less trivial approximation, Bob would
have to always output numbers strictly less than $1$.

A similar communication in expectation model is considered in \cite{FaenzaFioriniGrappeTiwary11},
where they show that the communication complexity is equal to the
logarithm of the non-negative rank (up to an additive constant) of the
associated slack matrix.
There is an important distinction, however; their model involves
communicating a slack matrix in expectation (the value $c-\inst(x)$), while
the model here deals directly with the underlying
combinatorial problem (the value $\inst(x)$).

\subsection{Sherali--Adams LP relaxations for CSPs}
\label{sec:Sherali--Adams-lp}

\newcommand{\lef}{\textrm{$\ell$.e.f.}\xspace}

A primary component of our approach involves leveraging
known integrality gaps for the Sherali--Adams (SA) hierarchy.
To that end, we now give a brief overview of Sherali--Adams LP relaxations.
For a more detailed account, we refer the reader to
\cite{laurent2003comparison}.

A $d$-round Sherali--Adams LP relaxation for a $\maxpi_n$ instance will consist
of variables $\{X_S : S \subseteq [n], |S| \leq d\}$ for all products
of up to degree-$d$ on the $n$ variables.  These variables $\{X_S : |S| \leq d\}$ are
to be thought of as the moments up to degree-$d$ of the variables,
under a purported distribution.

An important property of an SA solution $\{X_S : |S| \leq d\}$ is that
these moments agree with a set of local marginal distributions.  In particular, for every set $S
\sse [n]$ with $|S| \leq d$ there exists a distribution $\mu_S$ over
$\sbits^S$ such that,
$$ \E_{x \sim \mu_S} \chi_A(x) = X_A \qquad \forall A \sse S\,.$$

In an alternate but equivalent terminology, a $d$-round SA instance
can be thought of as \emph{$d$-local expectation functional
($d$-\lef)}.  Fix $n \geq 1$. We define
a \emph{$d$-local expectation functional} to be a linear
functional $\pE$ on degree-$d$ $n$-variate multilinear polynomials
such that $\pE \Ind =1$ and
$\pE P\ge 0$ for every degree-$d$ multilinear polynomial $P$ that is nonnegative
over $\sbits^n$ and depends only on $d$ variables.  In terms of the
local marginal distributions, $\pE : \{-1,1\}^n \to \R$ is the unique linear functional
satisfying
\begin{equation}\label{eq:moment}
 \pE \chi_{S} = \E_{x \sim \mu_S} \chi_S(x) \qquad \forall |S|
\leq d, S \sse [n] \,,
\end{equation}
and $\pE \chi_S = 0$ for $|S| > d, S \subseteq [n]$.

The \emph{$d$-round Sherali--Adams value} of a $\maxpi_n$ instance
$\inst$ is defined as
\begin{equation}\label{eq:sa-opt}
  \sa_d(\inst) \defeq \max_{\text{$d$-\lef $\pE$}} \pE \inst\mper
\end{equation}
This optimization problem can be implemented by an $n^{O(d)}$-sized
linear programming relaxation for $\maxpi_n$.
(Notice that $\pE$ is a $\left(\sum_{i=0}^{d} {n \choose i}\right)$-dimensional object.)
In particular, if $d$-rounds of Sherali--Adams achieve a
$(c,s)$-approximation for $\maxpi_n$, then so do general $n^{O(d)}$-sized LP
relaxations.

%
%
%

%

We remark that a $d$-\lef $\pE$ is a linear functional, but using self-duality of $L^2(\{-1,1\}^n)$,
we may also think of $\pE \in L^2(\{-1,1\}^n)$.  It has the Fourier representation
\[
\pE = \sum_{|\alpha| \leq d} \pE(\chi_{\alpha}) \chi_{\alpha}\,.
\]
We will use this representation freely.

\begin{lemma}\label{lem:saprops}
If $\pE$ is a $d$-\lef on $L^2(\{-1,1\}^n)$, then the following properties hold:
\begin{enumerate}[i]
\item For any non-negative $d$-junta $f : \sbits^n \to \mathbb R_+$, we have $\pE f \geq 0$.
\item For any $\alpha \subseteq [n]$, we have $|\pE \chi_{\alpha}| \leq 1$.
\item $\|\pE\|_{\infty} \leq \sum_{i=0}^{d} {n \choose i}\,.$
\end{enumerate}
\end{lemma}

\begin{proof}
Property (i) follows directly from the definition of a $d$-\lef.  Property (ii) follows from
\eqref{eq:moment}.  Property (iii) follows from (ii) using the fact that $\pE$
has at most $\sum_{i=0}^{d} {n \choose i}$ non-zero Fourier coefficients.
\end{proof}

It might help the reader, at this point to recall \pref{thm:lp-characterization} and the
representation \eqref{eq:instrep}.  Suppose that we had such a representation where the family
of functions $\{q_i\}$ were all $d$-juntas.  Fix an instance $\inst$ and let
$\pE$ denote an optimal solution to \eqref{eq:sa-opt}.
Applying $\pE$ to the right-hand side of \eqref{eq:instrep} yields
\[
\pE \left[\lambda_0 + \lambda_1 q_1 + \cdots + \lambda_R q_R\right] = \lambda_0 + \sum_{i=1}^R \lambda_i \pE \left[q_i\right] \geq 0\,,
\]
using \pref{lem:saprops}(i).
On the other hand, applying it to the left-hand-side yields
$\pE(c-\inst)=c-\sa_d(\inst)$.  Altogether, we conclude that
$\sa_d(\inst) \leq c$.  In particular, this holds for any $c \geq \mathcal L(\inst)$,
hence $\sa_d(\inst) \leq \mathcal L(\inst)$, implying that
in this special case (when all the $q_i$ functions are $d$-juntas),
the Sherali--Adams relaxation is at least as good as the given LP.
In general, our approach will be to {\em approximate} the $\{q_i\}$ functions
by juntas, and then apply a variant of this reasoning.

\begin{remark}
Some work on Sherali--Adams relaxations for \maxcut focus on edge variables
instead of vertex variables.  This includes
\cite{FernandezKenyonMathieu07,charikar2009integrality}.
In those papers, the $d$-round Sherali--Adams relaxation consists of variables
$\{X_S : S \subseteq {[n] \choose 2}, |S| \leq d\}$ for every subset of $d$ edges
in the complete graph.  Since their base polytope also includes triangle inequalities,
any ${d \choose 2}$-round Sherali--Adams solution with edge variables can be
converted to a $d$-round solution for vertex variables.
One should observe that the $d$-round vertex relaxation is
at least as strong as the $d$-round edge relaxation.

Moreover, both papers \cite{FernandezKenyonMathieu07,charikar2009integrality}
actually prove a lower bound against the $d$-round vertex version
and then argue that this yields a lower bound for the weaker edge relaxation.
For general max-CSPs, the vertex version is
arguably the canonical relaxation, and it is perhaps misguided
to consider the edge version even for \maxcut.
In \cite{schoenebeck2008linear} (which studies general CSPs), the more natural vertex version is considered.

A major benefit of the ``extended formulation'' model to
which our results apply
is that the edge/vertex relaxation distinctions are not relevant;
in fact no specific meaning is ascribed to the variables of the LP.  All that
matters is the number of defining inequalities.
\end{remark}


%

\section{Sherali--Adams and general LPs}
\label{sec:LPs}

Our main theorem is that general LP relaxations are no more powerful
than Sherali--Adams relaxations (in the polynomial-size regime).
%
%


\begin{theorem}[Main]
  \label{thm:main}
    Fix a positive number $d \in \mathbb N$,
    and a sequence of $k$-ary CSPs $\{\maxpi_n\}$, with $k \leq d$.
  Suppose that the $d$-round Sherali--Adams relaxation cannot achieve a $(c,s)$-approximation for $\maxpi_n$
  for every $n$.
  Then no sequence of LP relaxations of size at most $n^{d/2}$ can achieve a $(c,s)$-approximation
  for $\maxpi_n$ for every $n$.
\end{theorem}
%

We prove the following result for super-polynomial sized linear
programs in \pref{sec:putt-things-togeth}.

\begin{theorem}
  \label{thm:aux}
  Consider a function $f : \mathbb N \to \mathbb N$.
  Suppose that the $f(n)$-round Sherali--Adams relaxation cannot achieve a $(c,s)$-approximation for $\maxpi_n$.
  %
  Then for all sufficiently large $n$,
  no LP
  relaxation of size at most $n^{f(n)^2}$  can achieve a $(c,s)$-approximation for $\maxpi_N$, where $N \leq n^{10 f(n)}$.
\end{theorem}

In particular, by choosing $f(n) \asymp n^{\eps}$ for $\eps > 0$, and $n \asymp \left((\log N)/(\log \log N)\right)^{1/\eps}$,
known Sherali--Adams gaps for \maxcut \cite{charikar2009integrality} and \maxtwosat, \maxthreesat \cite{schoenebeck2008linear}
imply the same integrality gaps for LPs of size $n^{o( \frac{\log n}{\log \log n})}$.

\subsection{High-entropy distributions vs. juntas}
\label{sec:high-entr-distr}

Our first goal is to observe the following consequence
of Chang's Lemma \cite{Chang02} (and, specifically, the proof in \cite{IMR14}).

\begin{lemma}
  \label{lem:entropy}
  Let $q\from \sbits^n\to \Rnn$ be a density and let $\mu_q$
  denote the corresponding measure on $\sbits^n$.
  If $\mu_q$ has entropy
  at least $n-t$ for some $t \leq n$, then for every $1 \leq d \leq n$
  and $\gamma > 0$, there exists a set $J\sse [n]$ with
  \begin{equation}\label{eq:cardJ}\card{J}\le \frac{2td}{\gamma^2}\end{equation} such that for all subsets $\alpha\nsubseteq J$ with
  $\card{\alpha}\le d$, we have $\abs{\hat q(\alpha)} \le \gamma$.
\end{lemma}

\begin{proof}
Consider some $\gamma > 0$ and
let $S = \{ |\alpha| \leq d : |\hat q(\alpha)| > \gamma \}$.
Let $S' \subseteq S$ denote a maximal set of linearly independent
elements over $\mathbb F_2^n$.
In \cite{IMR14}, it is proved that
$|S'| \leq 2 \gamma^{-2} t$.
Let $J = \bigcup_{\alpha \in S'} \alpha$ so that $|J| \leq 2d\gamma^{-2} t$.
\end{proof}

\begin{remark}
	The claim in \cite{IMR14} (namely Lemma $2$ in \cite{IMR14}) is only stated for a $q=(2^n/|A|) {\bm 1}_{A}$ that
is the (scaled) characteristic
function of a subset $A \subseteq \{-1,1\}^n$, but the proof only uses the entropy of $q$.
A formal statement with a somewhat different proof can be found in \cite[\S 7]{LRS15}.
\end{remark}

\paragraph{Discussion of \pref{lem:entropy}}

It is interesting to note examples for which \pref{lem:entropy} cannot be improved much.
First, suppose that $n$ is odd, and consider the density coming from majority on $n$ bits:
\begin{equation}\label{eq:maj}
q(x) = 2 \cdot \bm{1}_{\{x_1 + \cdots + x_n > 0\}}\,.
\end{equation}
The corresponding measure $\mu_q$ has entropy $n-1$.
In this case, we have $|\hat q(\alpha)| \approx n^{-d/2}$ for $|\alpha|=d$, $d$ odd, and $d \ll\sqrt{n}$.
Thus \eqref{eq:cardJ} is essentially tight for $t=d=1$.

Consider the task of obtaining $|J| = n^{1-\delta}$ and $\gamma = n^{-\omega(1)}$, for some $\delta > 0$.
This is the interesting range of parameters in the next section.  For the majority density \eqref{eq:maj}, this is clearly impossible
in light of our discussion.
On the other hand, if one could obtain a rate of decay of the form $n^{-c(d)}$, with $c(d) \to \infty$ as $d \to \infty$
on the non-junta low-degree Fourier coefficients, then one could improve our main theorem (see \eqref{eq:error-terms}).

Unfortunately, the next example shows that this is impossible.
Let $k,n \in \mathbb N$ be such that $k$ divides $n$, and partition $\{1,2,\ldots,n\} = B_1 \cup B_2 \cup \cdots \cup B_{n/k}$
into $n/k$ disjoint blocks, each of size $k$.  Consider the density
\[
q(x) = \frac{k 2^k}{n} \sum_{i=1}^{n/k} {\bm 1}_{\{x_j = 1 \forall j \in B_i\}}\,.
\]
This function has a transitive symmetry, and thus for $k=o(n)$, does not admit an interesting junta set of size $o(n)$.
On the other hand for any $\alpha \subseteq B_i$, we have
\[
|\hat q(\alpha)| = |\langle q,\chi_{\alpha}\rangle| = \frac{k}{n} 2^{k-|\alpha|}\,.
\]
If we put $k=\sqrt{n}$, then we do not have an appreciable decay of the form $n^{-c(d)}$ with $c(d) \to \infty$ as $d \to \infty$.

But not all hope is lost:  It is plainly clear that $q$ can be approximated by a {\em non-negative combination of non-negative $k$-juntas.}
Furthermore, an approximation of this form would be just as good for us in the arguments that follow.
Thus another possible direction for improving our lower bounds significantly would be to prove a variant of
\pref{lem:entropy}
using an approximation by convex combinations of non-negative juntas, such that one achieves a strong form
of decay on the Fourier coefficients.

Some improvement is possible in this case:  In the setting of
\pref{lem:entropy}, one can achieve a non-negative combination of $k$-juntas
with $k=O(td/\gamma)$ (as opposed to $\gamma^2$); see \cite[\S 7]{LRS15}.
But this approach too reaches a bottleneck:
Suppose that $m$ divides $n$ and partition $[n] = S_1 \cup S_2 \cup \cdots \cup S_{n/m}$
where $|S_i|=m$.
Consider
functions of the form $q(x)=f\left(\chi_{S_1}(x), \chi_{S_2}(x), \ldots, \chi_{S_{n/m}}(x)\right)$
where $\chi_S(x)=\prod_{i \in S} x_i$ is the corresponding Fourier
character and $f: \sbits^{n/m} \to \R_{\geq 0}$ is a function.  The effect of this
operation is to lift the low-degree Fourier coefficients of $f$ to
higher-degree coefficients of $q$,
cutting off the hope for a strong form of decay.
For instance, if $m=\sqrt{n}$ and $f$ is the majority
density on $n/m=\sqrt{n}$ bits (as in \eqref{eq:maj}).

\subsection{Random restrictions}
\label{sec:rand-restr-proof}

  %
  %
  %
  %
  %

We first recall the following standard estimates (see, e.g., \cite{McD98}).
	Suppose $X_1,\ldots, X_n$ are i.i.d $\{0,1\}$ random variables
	with $\E[X_i] = p$.  Then,
\begin{equation}\label{eq:chernoff}
\Pr\left(\sum_{i=1}^n X_i \geq \frac{pn}{2}\right) \geq 1- e^{-pn/8}\,.
\end{equation}
Furthermore,
\begin{equation}\label{eq:easy}
	\Pr\left(\sum_{i=1}^n X_i \geq t\right) \leq \sum_{S \in
		\binom{n}{t}} \Pr\left( \sum_{i \in S} X_i = t \right)
		\leq \binom{n}{t} \cdot p^t \leq (pn)^t
\end{equation}

%

\begin{lemma} \label{lem:random-restriction}
For any $d \in \mathbb N$, the following holds.
Let $Q$ be a collection of densities $q\from \sbits^n\to \Rnn$
such that the corresponding measures $\mu_q$ have entropy at least $n-t$.
If $\card{Q} \leq n^{d/2}$, then
  for all integers $m$ with $3 \leq m \leq n/4$, there exists a set $S\sse [n]$ such
  that
\begin{itemize}
	\item $|S|=m$

	\item  For each $q \in Q$, there exists a set of at most $d$
		coordinates $J(q) \sse S$ such that under
		the distribution $\mu_q$, all $d$-wise
		correlations in $S - J(q)$ are small.  Quantitatively, we have
		$$ \abs{\hat{q}({\alpha})} \leq \left(\frac{16
		mtd}{\sqrt{n}}\right)^{1/2} \qquad
		\forall \alpha \sse S, \alpha \not\subseteq J(q),
		|\alpha| \leq d $$
\end{itemize}
\end{lemma}
\begin{proof}
	We will sample the set $S \sse [n]$ by including each element
	independently with probability $2m/n$,
then argue that with non-zero probability, both the conditions on $S$
hold.

	First, by \eqref{eq:chernoff}, we have $|S| \geq m$
	with probability at least $1-e^{-m/4} > 1/2$.

	Fix $\gamma = \left(\frac{16 mtd}{\sqrt{n}}
	\right)^{\nfrac{1}{2}}$.
	By \prettyref{lem:entropy}, for each $q \in Q$ there exists a
	set $J'(q)$ of at most $\frac{2td}{\gamma^2} \leq \frac{\sqrt{n}}{8m}$ coordinates such
	that for all subsets $\alpha \nsubseteq J'(q)$ with $|\alpha| \leq d$, we have $|\hat q(\alpha)| \leq \gamma$.

	The set $J(q)$ for a distribution $q$ is given by $J(q) =
	J'(q) \cap S$.
	By \eqref{eq:easy},
    we can write
	$$ \Pr{ \left[|J'(q) \cap S| \geq  d
	\right]} \leq  \left(\frac{2m}{n} \cdot |J'(q)|\right)^{d} \leq
	\left(\frac{2m}{n} \cdot \frac{\sqrt{n}}{8 m}
	\right)^{d} \leq \frac{1}{4^d n^{\nfrac{d}{2}}}\,.$$
	The existence of the set $S$ follows by taking a union bound over all the $|Q| \leq
	n^{\nfrac{d}{2}}$ densities in the family $Q$.  Note that we have concluded with $|S| \geq m$,
    but we can remove some elements from $S$ to achieve $|S|=m$.
\end{proof}
%

\subsection{Proof of Main Theorem}
\label{sec:putt-things-togeth}

In this subsection, we will prove \pref{thm:main} and \pref{thm:aux}.  Let $m \leq n$ be parameters $m,n \in \mathbb N$ to be chosen later.
Consider an instance $\inst_0$ of $\maxpi_m$.
Recalling \eqref{eq:sa-opt}, let $\pE$ be a corresponding optimal $d$-\lef, i.e. such that $\pE[\inst_0]=\sa_d(\inst_0)$.

Suppose that $\mathcal L$ is an LP relaxation of size at most $R \leq n^{d/2}$ for $\maxpi_n$.  Our goal is to show that
there exists an instance $\inst$ that is a ``shift'' of $\inst_0$, and a
value $\e_n > 0$ such that $\mathcal L(\inst) \geq \sa_d(\inst_0) - \e_n$, with $\e_n \to 0$ as $n \to \infty$.
By ``shift,'' we mean a planting of the instance $\inst_0$ on some subset of the variables $\{1,2,\ldots,n\}$.
Since $\opt(\inst)=\opt(\inst_0)$, we will conclude our proof by taking $\e_n \to 0$.

\medskip

By \pref{thm:lp-characterization}, there are densities $q_1, q_2, \ldots, q_R : \sbits^n \to \Rnn$
such that for every $\maxpi_n$ instance $\inst$, we have
\begin{equation}\label{eq:factor}
\mathcal L(\inst) - \inst = \lambda_0(\inst) + \sum_{i=1}^R \lambda_i(\inst) q_i\,,
\end{equation}
for some non-negative numbers $\lambda_i(\inst)$ depending on $\inst$.
For some $t \geq 0$ to be chosen later,
let \[Q_t = \{ 1 \leq i \leq R : \|q_i\|_{\infty} \leq 2^t \}\,.\]
Observing that the left-hand side of \eqref{eq:factor} is pointwise at most 1, for any $i \notin Q_t$, we must have $\lambda_i(\inst) \leq 2^{-t}$ for every instance
$\inst$.
At this point, one should also observe that $\sum_{i=0}^R \lambda_i(\inst)\leq 1$ by taking expectations over both sides of \eqref{eq:factor}.

If $i \in Q_t$, then since $\|q_i\|_{\infty} \leq 2^t$, we can lower
bound the entropy of $\mu_{q_i}$
as follows,
$$H(\mu_{q_i}) = \sum_{x \in \sbits^n} \frac{q_i (x)}{2^n} \log
\frac{2^n}{q_i(x)} \geq \left(\sum_{x \in \sbits^n} \frac{q_i(x)}{2^n}
\right) \cdot  \log
\frac{2^n}{\|q_i\|_{\infty}} = \log \frac{2^n}{\|q_i \|_{\infty}} \geq
n-t \mper$$
Apply \pref{lem:random-restriction} to the set of densities with index
in $Q_t$, and let $S \subseteq [n]$
with $|S|=m$ be the subset whose existence is guaranteed.  Let $\inst_S$ denote the instance $\inst_0$ planted
on the subset $S$, and similarly let $\pE_S$ be the Sherali--Adams functional
$\pE$ planted on $S$.  Equation \eqref{eq:factor} gives us a representation of the form
\begin{equation}\label{eq:factor2}
\mathcal L(\inst_S) - \inst_S = \lambda_0(\inst_S) + \sum_{i=1}^R \lambda_i(\inst_S) q_i\,,
\end{equation}

For each $i \in Q_t$, let $q^S_i = \sum_{\alpha \subseteq S} \hat q_i(\alpha) \chi_{\alpha}$.  Observe that $q^S_i$
is the conditional density on the variables in $S$ (equivalently, we obtain $q^S_i$ by averaging
over all variables outside $S$).  By our application of \pref{lem:random-restriction}, we can write
$q^S_i = \tilde q^S_i + e_i$ where $\tilde q^S_i=\sum_{\alpha \subseteq J(q_i) \cap S} \hat q_i(\alpha) \chi_{\alpha}$
is a non-negative $d$-junta and $|\hat e_i(\alpha)| \leq \left(\frac{16mtd}{\sqrt{n}}\right)^{1/2}$
for all $|\alpha| \leq d$.

Using the fact that $\pE_S$ only depends on variables in $S$, we have $\pE_S(q_i)=\pE_S(q^S_i)$ for all $i \in Q_t$.
Also observe that for $i \notin Q_t$, we have $|\tilde {\mathbb E}_S(q_i)| \leq \|\tilde {\mathbb E}_S\|_{\infty} \leq {m \choose d}$,
first using $\E q_i=1$ and then using property (iii) of \pref{lem:saprops}.
Now we apply $\pE_S$ to both sides of \eqref{eq:factor2} to obtain
\begin{eqnarray*}
\mathcal L(\inst_S) - \sa_d(\inst_0) &=& \lambda_0(\inst_S) + \sum_{i \in Q_t} \lambda_i(\inst_S) \tilde {\mathbb E}_S(q^S_i)
+ \sum_{i \notin Q_t} \lambda_i(\inst_S) \tilde {\mathbb E}_S(q_i) \\
&\geq &
\sum_{i \in Q_t} \lambda_i(\inst_S) \left(\tilde {\mathbb E}_S(\tilde q^S_i) + \tilde {\mathbb E}_S(e_i)\right)
- \|\tilde {\mathbb E}_S\|_{\infty} \sum_{i \notin Q_t} \lambda_i(\inst_S) \\
&\geq &
\sum_{i \in Q_t} \lambda_i(\inst_S) \tilde {\mathbb E}_S(e_i)
- {m \choose d} n^{d/2} 2^{-t}\,,
\end{eqnarray*}
where in the final line we have used the fact that $\tilde q_i^S$ is a non-negative $d$-junta (along with property (i) of \pref{lem:saprops}),
the fact that $\lambda_i(\inst_S) \leq 2^{-t}$ for $i \notin Q_t$,
and our assumption that the total number of indices $i \in \{1,2,\ldots,R\}$ is at most $n^{d/2}$.

Finally, it remains to observe that \begin{equation}\label{eq:error-terms}
|\tilde {\mathbb E}_S(e_i)| \leq \sum_{\alpha \subseteq S} |\tilde {\mathbb E}_S [\chi_{\alpha}]| \cdot |\hat e_i(\alpha)|
\leq {m \choose d} \left(\frac{16mtd}{\sqrt{n}}\right)^{1/2}\,,\end{equation}
where we have employed property (ii) of \pref{lem:saprops}.
Plugging this estimate into the preceding inequality yields
\[
\mathcal L(\inst_S) - \sa_d(\inst_0) \geq -{m \choose d} \left(\frac{16mtd}{\sqrt{n}}\right)^{1/2} - {m \choose d} n^{d/2} 2^{-t}\,.
\]
If we set $t=d \log_2 n$, then $\mathcal L(\inst_S)\geq  \sa_d(\inst_0) -\e_n$, where
\begin{equation}
\label{eq:eta}
\e_n = O\left(\frac{m^d \sqrt{md \log n}}{n^{1/4}}\right)\,.
\end{equation}
Clearly for $m,d$ fixed, we have $\e_n \to 0$ as $n \to \infty$, completing the proof
of \pref{thm:main}.
\qed

\begin{proof}[Proof of \pref{thm:aux}]
	Fix an instance size $m$ and put $d = f(m)$.  In the preceding argument,
    require that $n$ grows like $m^{10d} = m^{10 f(m)}$ so that $\e_n = o(1)$ (see \eqref{eq:eta}).
    The lower bound achieved is $n^{d/2} \geq m^{5 f(m)^2}$.
\end{proof}

\subsection{Nonnegative rank}
\label{sec:smoothrank}

The lower bounds of \pref{thm:main}
can be stated equivalently in terms of nonnegative rank.
We recall that the nonnegative rank of a nonnegative matrix $A \in \mathbb R_+^{m \times n}$
is defined by
\[
\rank_+(A) = \min \left\{r : A_{ij} = \langle u_i, v_j\rangle \textrm{ for some } \{u_i, v_j\} \subseteq \mathbb R_+^r \right\}\,.
\]
Fix $n \geq 1$.
Let \(M=(M_{G,x})\) be the matrix indexed by \(n\)-vertex \maxcut
instances with \maxcut value at most \(s\)
(e.g., \(s=1/2+\gamma\)) and bipartitions \(x\in\{\pm 1\}^n\) such that \[
M_{G,x} = c - G(x)
\,,
\]
where $G(x)$ denotes the fraction of edges crossing the bipartition corresponding to $x$.
A corollary of \pref{thm:main} is that $\rank_+(M) \geq n^{\Omega\left(\frac{\log n}{\log \log n}\right)}$.

Define for $\e > 0$, the {\em \(\e\)-smooth nonnegative rank of a matrix} \(A\) as \[
{\mathrm{rank}_{+,\e}}(A) \seteq \min  \left\{ {\mathrm{rank}_{+}}(A') \mid (1-\e) A_{ij} \le A'_{ij} \le (1+\e)A_{ij} \,\,\forall i,j \right\}
\mcom
\]

Our main result shows that \(M\) has superpolynomial nonnegative rank.
We claim that \(M\) has only polynomial nonnegative approximate rank
for every fixed $\eps > 0$.
(Since the entries of \(M\) are bounded from above and bounded away from \(0\), the notions of approximate and smooth nonnegative rank coincide.)
In order to demonstrate that \(M\) has small approximate nonnegative
rank, for each $t \in \N$, we will exhibit a matrix $M'$ that approximates $M$ well,
$\left| M'_{G,x} - M_{G,x} \right| \leq 2^{-\Omega(t)}$ for all $G$ and
$x$, but has a small nonnegative rank, i.e., $\mathrm{rank}_{+}(M')
\leq n^{O(t)}$.
To this end, we will use the reformulation of nonnegative rank as a communication
model discussed in \pref{sec:background}.  Consider the following communication protocol between Alice and Bob:
\begin{itemize}
\itemsep1pt\parskip0pt\parsep0pt
\item
  Alice receives as input an \(n\)-vertex graph \(G\) with \maxcut value at most \(s\).
\item
  Bob receives as input a bipartition \(x\in\{\pm 1\}^n\).
\item
  Alice chooses \(t\) edges \(e_1,\ldots,e_t\) of \(G\) independently at random and sends the endpoints of the sampled edges to Bob.
\item
  Bob computes what fraction \(\theta\) of the edges \(e_1,\ldots,e_t\) cross the bipartition \(x\).
  If \(\theta > c\), then Bob outputs \(0\) else Bob outputs \(c-\theta\).
\end{itemize}

Let $M'$ be the matrix computed by the above protocol, i.e.,
$M'_{G,x}$ is the expected value of Bob's output when Alice receives
the graph $G$ and Bob receives $x$ as input.
This protocol yields a rank-\(2^{t\cdot 2\log n}\) nonnegative factorization of the matrix \(M'\),
as follows.

Let \(\theta_{G,x}\) be the random variable given by Bob's output when the inputs for Alice and Bob are \(G\) and \(x\), respectively.
Then, \(M'_{G,x}=c-\E \theta_{G,x}\) for all \(G\) and \(x\).

At the same time, we have
\[
M_{G,x}=c-\E \left[\theta_{G,x}\right] - \delta_{G,x}\,,
\]
where
\[
\delta_{G,x} \leq 2 \Pr[\theta_{G,x} > c]\,.
\]

In words, the discrepancy between the computed matrix \(M'\) and the target matrix \(M\) is accounted for by the probability of the events \(\{\theta_{G,x}>c\}\).
However, since \(G(x)\le s <c\) is bounded away from \(c\) by some
constant,
a standard Chernoff bound yields
\[\Pr[\theta_{G,x}> c]\le 2^{-\Omega(t)}\,,\] for all \(G\) and \(x\).
It follows that the matrix \(M'\) satisfies \({\mathrm{rank}_{+}}(M')\le n^{O(t)}\) and \(\max_{G,x} \lvert M'_{G,x} - M_{G,x} \rvert\le 2^{-\Omega(t)}\).

\Jnote{
\subsection{Extension to non-Boolean CSPs}
\label{sec:non-Boolean}
}

 \section{Symmetric linear programs}
 \label{sec:symm-formulations}

We will now prove the following theorem relating Sherali--Adams gaps to those
for symmetric LPs for \maxcut.  While this connection holds more generally for max-CSP problems,
we will focus on \maxcut for clarity.
Recent work has extended these ideas to problems like TSP \cite{LRST14}, and also to
a connection between symmetric SDPs and the Sum-of-Squares hierarchy \cite{LRST14,FSP14}.

\begin{theorem}
\label{thm:symmetric}
Fix a $k$-ary CSP $\maxpi$ over the boolean domain.
Suppose that, for some numbers $m, d > 0$, the $d$-round Sherali--Adams relaxation
for $\maxpi_m$ cannot achieve a $(c,s)$-approximation.
Then no symmetric LP of size $\leq \binom nd$ can achieve a
$(c,s)$-approximation on $\maxpi_n$ where $n=2m$.
\end{theorem}
We note here that the Sherali-Adams hierarchy produces symmetric
linear programs.  Hence, the above result can be viewed as asserting
that Sherali-Adams hierarchy is {\it complete} for the class of
symmetric linear programs.

By appealing to the known Sherali--Adams gaps for \maxcut
\cite{charikar2009integrality} and \maxtwosat, \maxthreesat
\cite{schoenebeck2008linear}, we get the same integrality gaps for
arbitrary symmetric LPs.  For example, in the case of \maxcut, we
obtain the following lower bound.
\begin{corollary}
For every $\eps > 0$, there exists $\delta > 0$ such that no
symmetric linear program of size $2^{n^{\delta}}$ yields a $(1-\eps,
\frac{1}{2}+\eps)$-approximation for \maxcut.
\end{corollary}

In order to prove \prettyref{thm:symmetric}, we will need the
following characterization of symmetric function families.
\begin{lemma} \label{lem:symmetry-family}
	Suppose a family of functions $Q = \{q_i : \sbits^n \to
	\mathbb R : i=1,2,\ldots,R\}$ is
	closed under the action of $\Sym(n)$.
 If $R < \binom{n}{d}$
	for $d < n/4$, then each function $q_i$ depends only on a
	subset $J_i \subseteq [n]$ of at most $d$ coordinates and
	possibly the value  of the sum $\sum_{j=1}^n
	x_j$.
\end{lemma}
\begin{proof}
Here we will need a few basic notions about group actions.
A group $\cG$ {\it acts} on a universe $\cX$, if each element $g
\in \cG$ {\it permutes} the elements of the universe $\cX$, and this
action commutes with the group operation.  Formally,
a group action is defined by a map $\iota: \cG \times \cX \to \cX$ such that
$\iota(g,\iota(h,x)) = \iota(gh, x)$ for all $g,h \in \cG$ and $x \in
\cX$.  For convenience, we will denote $g \cdot x \defeq \iota(g,x)$.

For an element $x \in \cX$, its orbit $\orb(x)$ is given by $\orb(x) = \{ g
\cdot x \mid g \in \cG\}$ and its stabilizer is given by $\Stab(x) =
\{g \in \cG \mid g\cdot x = x \}$.  A basic fact from group theory is
that for every action of a finite group $\cG$ and every $x \in X$,
it holds that $| \Stab(x) | \cdot |
\orb(x)| = | \cG|$.

The group $\Sym(n)$ of all permutations on $n$ elements acts on the
space of functions over $\sbits^n$ by permutation of the coordinates.
Let $\orb(f)$ denote the orbit of a function $f$ under the
action of $\Sym(n)$, and let $\Stab(f)$ denote the
stabilizer of $f$.  Since $Q$ is closed under this action,
it contains the orbits of each of the functions $q_1,
\ldots, q_R$.

This implies that for each $i \in [R]$, we have
$|\orb(q_i)| <
\binom{n}{d}$.  Since $| \orb(q_i) | \cdot | \Stab(q_i) | =
| \Sym(n)| = n!$, we conclude that for each $i\in [R]$,
$|\Stab(q_i)| \geq d!(n-d)!$.
At this point, we appeal to the following group theoretic fact that we
borrow from the work of Yannakakis \cite{Yannakakis91}.

\begin{lemma}[{\cite[Claim~2]{Yannakakis91}}]
	\label{lem:grouptheory}
Let $H$ be a group of permutations whose index in $\Sym(n)$ is at most
$\binom{n}{d}$ for some $d < n/4$.  Then there exists a set $J \subseteq \{1,2,\ldots,n\}$ of size at
most $d$ such that $H$ contains all even permutations that fix
the elements of $J$.
\end{lemma}

By \prettyref{lem:grouptheory}, the stabilizer subgroup
$\Stab(q_i)$ contains all even permutations that fix a subset of
coordinates $J_i$ with $|J_i| \leq d$.
We claim that $\Stab(q_i)$ contains all permutations that fix the
coordinates in $J_i$.  We know that for every $x \in \sbits^n$, and every even permutation $\sigma \in \Sym(\bar J_i)$
we have, $q_i(x) = q_i(\sigma x)$.
Here, we use $\Sym(\bar J_i)$ to denote the subgroup of $\Sym(n)$ fixing elements in $J_i$.

For every $x \in \sbits^n$, there will be
two coordinates $a,b \in \bar{J_i}$ such that $x_a = x_b$.  Let
$\pi_{ab}$ denote the transposition that swaps $a$ and $b$.   Since
$\pi_{ab}(x) = x$, we have $q_i(\pi_{ab}(x)) = q_i(x)$.  So for even permutations $\sigma \in \Sym(\bar J_i)$,
$$ q_i(\sigma \pi_{ab} x) = q_i(\pi_{ab} x) = q_i(x)\,.$$
As $\sigma$ varies over all even permutations in $\Sym(\bar J_i)$, $\sigma \pi_{ab}$
varies over all odd permutations in $\Sym(\bar J_i)$, leading to the conclusion that
$\Sym(\bar J_i) \subseteq \Stab(q_i)$.

This symmetry of the function $q_i(x)$ implies that it depends only on
the assignment to coordinates in $J_i$ and the hamming weight of the assignment to
coordinates in $\bar{J_i}$, i.e. the value $\sum_{j=1}^n x_j - \sum_{j \in J_i} x_j$.
This shows that $q_i$
is a function
depending only on the coordinates in $J_i$ and the value $\sum_{i \in [n]} x_i$.
\end{proof}

We are now in position to prove the main theorem of this section.

\begin{proof}[Proof of \pref{thm:symmetric}]
	Let $\cL$ be a symmetric LP relaxation for $\maxpi_n$ of size
	$R \leq \binom{n}{d}$.
    Supposing that this relaxation achieves a $(c,s)$-approximation,
    we will derive a contradiction.

By
	applying \prettyref{thm:lp-characterization}, there exists a
	family of functions $Q = \{q_1,\ldots, q_R \from \sbits^n \to  \R_{\geq 0}\}$
	such that for every instance $\inst$ of
	$\maxpi_n$ with $\opt(\inst) \leq s$, we have
	$$ c - \inst = \lambda_0 + \sum_{i=1}^R \lambda_i q_i\,.$$
	for some non-negative weights $\{\lambda_i\}_{i=0}^R$.
	Moreover, the family of functions $Q$ is invariant under
	the action of $\Sym(n)$.
	Therefore, by
	\prettyref{lem:symmetry-family}, each of the functions $q_i \in Q$
	depends on a set $J_i$ of at most $d$ coordinates and possibly the value $\sum_{i=1}^n
	x_i$.

	Fix an instance $\inst$ of $\maxpi_m$ on which the $d$-round
	Sherali--Adams relaxation fails to achieve a
	$(c,s)$-approximation, i.e.,
	$$ \opt(\inst) \leq s \qquad \text{ and } \qquad \sa_d(\inst) >
	c\,.$$

	For $n = 2m$, construct an instance $\inst'$ of $\maxpi_n$ by including $m$
	additional dummy variables in $\inst$ with no constraints
	among them.  Concretely, if $X_1,\ldots, X_n$ are
	variables in $\inst'$, then restricted to the variables $X_1,
	\ldots, X_m$, the constraints are identical to $\inst$ while
	there are no constraints among $X_{m+1},\ldots, X_n$.

	For an assignment $x \in \sbits^n$, we will denote $x_A =
	(x_1,\ldots, x_m)$ and $x_B = (x_{m+1},\ldots, x_n)$.  In this
	notation, it is easy to see that for every assignment $x$,
	$$ \inst'(x) = \inst(x_A) \mper$$

	By construction, we have $\opt(\inst') = \opt(\inst) \leq s$.
	Since the symmetric LP relaxation $\cL$ yields a
	$(c,s)$-approximation to $\maxpi_n$, there exist $\{ \lambda_i \geq 0\}_{i=0}^R$ such that
	$$ c - \inst' = \lambda_0 + \sum_{i=1}^R \lambda_{i} q_i \mper$$
	Using $\inst(x) = \inst(x_A)$, we can rewrite the above
	identity as,
		$$ c - \inst(x_A) = \lambda_0 + \sum_{i=1}^R \lambda_i
		q_i(x_A,x_B) \mper$$
	Define $h_i:\sbits^m \to \R_{\geq 0}$ as $h_i(x) = q_i(x,-x)$.
	Setting $x_B = -x_A$ in the above identity, we arrive at
	\begin{equation} \label{eq:SAjunta}
	 c - \inst(x_A) =  \lambda_0 + \sum_i \lambda_i h_i(x_A) \mper
	\end{equation}
	Recall that each of the functions $q_i$ depends on a subset
	$J_i$ of at most $d$ coordinates and possibly the value of
	$\sum_{i=1}^n x_i$.  This implies that $h_i(x) = q_i(x,-x)$ is a
	$d$-junta, since the sum of all the coordinates of $(x,-x)$ is always
	equal to $0$.  In particular, the identity in
	\eqref{eq:SAjunta} expresses the function $c - \inst$ as a
	non-negative combination of $d$-juntas.

	Let $\pE$ denote the $d$-round Sherali--Adams functional for
	the instance $\inst$.  Apply the $\pE$ functional to both
	sides of \eqref{eq:SAjunta} to obtain a contradiction.  By
	\prettyref{lem:saprops} and the fact that each $h_i$ is a
	non-negative $d$-junta, we have $\pE[h_i] \geq 0$.
	On the other hand, the
	left hand side  $\pE (c - \inst) = c - \sa_d(\inst) < 0$.
\end{proof}

\section{Conclusion}
\label{sec:conclusion}

We have shown that for constraint satisfaction problems,
there is an intimate relationship between general polynomial-sized linear programs
and those arising from $O(1)$ rounds of the Sherali--Adams hierarchy.
There are a few natural questions that readily suggest themselves.

Firstly, our quantitative bounds are far from optimal.  For instance,
it is known that the integrality gap of $1/2+\varepsilon$ for \maxcut
persists for $n^{c_{\varepsilon}}$ rounds of Sherali-Adams hierarchy,
where $c_\varepsilon$ is some constant depending on $\varepsilon$ \cite{charikar2009integrality}, while
we are only able to prove an integrality gap for LPs of size $n^{o(\frac{\log n}{\log \log n})}$.
This is due to the factor of $m^d$ appearing in our Fourier estimate \eqref{eq:error-terms}.

\begin{question}\label{ques:improve}
Is it the case that for approximating (boolean) max-CSP problems on $n$ variables, linear programs of size $R(n)$ are only
as powerful as those arising from $\mathrm{poly}(\frac{\log {R(n)}}{\log n})$ rounds of the Sherali--Adams hierarchy?
\end{question}

Secondly, given the connection for linear programs, it is natural to suspect that a similar
phenomenon holds for SDPs.

\begin{question}
For max-CSP problems, is there a connection between the efficacy of general SDPs and those from the Sum-of-Squares SDP hierarchy \cite{Laserre01,Parrilo00}?
\end{question}

As mentioned in the introduction, recent work \cite{LRS15}
yields a positive solution to this question, although
the approach has similar limitations to those
highlighted in Question \ref{ques:improve}.

Finally, our techniques have made very strong use of the product structure on the space
of feasible assignments for CSPs.
One might hope to extend these connections to other types of problems
like TSP and finding maximum-weight perfect matchings in general graphs \cite{Rothvoss14,Yannakakis91} or approximations
for vertex cover.
See \cite{BPZ15} for progress on the latter problem.

\subsection*{Acknowledgements}

We thank the anonymous referees for many useful suggestions and observations.

S. O.~Chan was supported by NSF grants CCF-1118083 and CCF-1017403.
P.~Raghavendra was supported by NSF Career Award CCF-1343104 and an Alfred
P. Sloan Fellowship.  J.~R.~Lee was supported by NSF grants CCF-1217256 and CCF-0905626.
D.~Steurer was supported by NSF grants, an Alfred P. Sloan Fellowship, and a Microsoft Research Faculty Fellowship.

\addreferencesection
\bibliographystyle{amsalpha}
\bibliography{lpsize,afmed2013}

\newpage

\appendix

\section{What is Sherali--Adams?}\label{what-is-sheraliadams}

Our definition of Sherali--Adams relaxation differs from the definition in prior works (in particular, the works that proved lower bounds on the size of Sherali-Adams relaxations for approximating CSPs) \cite{SheraliAdams1990,FernandezKenyonMathieu07, charikar2009integrality}.
This discrepancy stems from the fact that traditionally LP hierarchies like Sherali--Adams are applied to integer linear programming formulations of a problem, whereas our relaxations can be viewed as applying the analogous reasoning to a more direct formulation of the problem.
It turns out that the latter approach typically leads to relaxations that are easier to describe and a-priori more powerful.

We will argue that the two versions of Sherali--Adams are equivalent for the problems we consider, in the sense that each relaxation in one hierarchy is captured by a relaxation in the other hierarchy of comparable size (at most a polynomial factor more constraints).

We remark that our relaxations are equivalent to the viewpoint of Sherali--Adams as a collection of mutually-consistent ``local distributions over assignments.''
\Dnote{TODO: maybe better to explain this term instead of using square quotes} This viewpoint was used in previous works for proving lower bounds.
These previous works show that this viewpoint captures the power of Sherali--Adams.
We will argue that this viewpoint is indeed equivalent to the Sherali--Adams hierarchy.

\subsection{Edge-based Sherali--Adams relaxations for \maxcut}\label{edge-based-sheraliadams-relaxations-for-max-cut}

The \emph{cut polytope} \({\mathrm{CUT}}_n\) is the convex hull of all vectors \(y\in \{ 0,1 \}^{\binom{n}{2}}\) such that there exists a bipartition \(x\in\{\pm 1\}^n\) with \(y_{i,j}=\Ind_{\{x_i\neq x_j\}}\) for all \(i\neq j\in [n]\).
We can formulate \maxcut as the problem of optimizing a linear function of the form \(\sum_{ij\in E(G)} y_{i,j}\) over \({\mathrm{CUT}}_n\) for a graph \(G\).
The standard LP relaxation of \({\mathrm{CUT}}_n\) is the \emph{metric polytope} \({\mathrm{METRIC}}_n\), which consists of all vectors \(y\in [0,1]^{\binom{n}{2}}\) that satisfy the inequalities \(y_{i,j} \leq y_{i,k} + y_{k,j}\) and \(y_{i,j} + y_{i,k} + y_{k,j} \leq 2\) for all \(i,j,k\in [n]\).
This \(O(n^3)\)-size LP relaxation corresponds to an exact integer linear programming (ILP) formulation in the sense that the convex hull of the integer vectors \({\mathrm{METRIC}}_n\cap \{0,1\}^{\binom n 2}\) is precisely the cut polytope \({\mathrm{CUT}}_n\).
In our notation, the level-\(r\) Sherali--Adams relaxation of this ILP formulation consists of all linear functionals \(\pE\) on \(L_2(\{0,1\}^{\binom n 2})\) such that \(\pE1=1\) and \(\pE f\cdot \ell \ge 0\) for every non-negative \(r\)-junta \(f\) on \(\{0,1\}^{\binom n 2}\) and every linear function \(\ell\) on \(\{0,1\}^{\binom n 2}\) corresponding to one of the defining linear inequalities of \({\mathrm{METRIC}}_n\), i.e., \(\ell\) is of the form \(y_{i,j}\), \(1-y_{i,j}\), \(y_{i,k} + y_{k,j} - y_{i,j}\), or \(2-y_{i,j}- y_{i,k} -y_{k,j}\) for some \(i,j,k\in [n]\).
The value of the level-\(r\) Sherali--Adams relaxation for a \maxcut instance \(G\) is the maximum value of \(\pE \sum_{ij\in E(G)} y_{i,j}\) over all linear functionals \(\pE\) that satisfy the previous conditions.
(From our description it is not immediately clear that this optimization problem has a small linear programming formulation.
However note that \(n^{O(r)}\) linear inequalities are enough to define the set of all admissible linear functionals \(\pE\).
Hence, we can reduce this problem to a linear program of size \(n^{O(r)}\).
It's also possible, but somewhat cumbersome, to describe this small linear program explicitly \cite{SheraliAdams1990, FernandezKenyonMathieu07, charikar2009integrality}.)

\subsubsection{Why is this hierarchy of relaxations equivalent to the previously described hierarchy?}\label{why-is-this-hierarchy-of-relaxations-equivalent-to-the-previously-described-hierarchy}

Let \(G\) be any graph.
First, consider any \(k\)-local pseudo-expectation \(\pE_x\) as defined before.
We will construct an equivalent linear functional \(\pE_y\) for the
level-\(r\) Sherali--Adams relaxation with \(r = k/2-3\).
Recall that \(\pE_x\) is a linear functional on \(L_2(\{\pm 1\}^n)\) such that \(\pE_x 1 =1\) and \(\pE_x f \ge 0\) for every nonnegative \(k\)-junta \(f\).
We define a linear functional \(\pE_y\) on \(L_2(\{0,1\}^{\binom n 2})\) as follows, \[
\pE\nolimits_y f = \pE\nolimits_x (f \circ \varphi)
\mcom
\] where \(\varphi\) is the function that maps any bipartition \(x\in\{\pm 1\}^n\) to the corresponding vector \(\{0,1\}^{\binom n 2}\), i.e., \(\varphi(x)_{i,j}=\Ind_{\{x_i\neq x_j\}}\).
(Note that algebraically \(\varphi(x)_{i,j}=(1-x_ix_j)/2\).)
This linear functional satisfies \(\pE_y1=\pE_x 1=1\) and \(\pE_y \sum_{i,j \in E(G)} y_{i,j}=\pE_x \sum_{ij \in E(G)} (1-x_ix_j)/2\).
Consider any nonnegative \(r\)-junta \(f\) over \(\{0,1\}^{\binom n 2}\) and any facet defining linear inequality \(\{\ell \ge 0\}\) for \({\mathrm{METRIC}}_n\).
We are to show \(\pE_y f \cdot \ell = \pE_x (f \circ \varphi) \cdot (\ell \circ \varphi)\ge 0\).
Since \(\{\ell\ge 0\}\) is a valid inequality for the vertices of \({\mathrm{CUT}}_n\), we have \(\ell \circ \varphi\ge 0\) over \(\{\pm 1\}^n\).
Therefore, \((f \circ \varphi) \cdot (\ell \circ \varphi)\) is nonnegative over \(\{\pm 1\}^n\).
Notice that each facet defining linear inequality \(\{\ell \ge 0\}\)
for \({\mathrm{METRIC}}_n\) depends only on three $y_{ij}$ variables,
and therefore $\ell \circ \phi$ depends on at most $6$ of the
variables $\{x_1,\ldots, x_n\}$.  Therefore, the function $(f \circ
\phi) \cdot (\ell \circ \phi)$ depends on at most \(2r+6\le k\) of the
variables $\{x_1,\ldots,x_n\}$.
It follows that \(\pE_x (f\circ \varphi) \cdot (\ell \circ \varphi)\ge 0\) as required.

Next, consider any linear functional \(\pE_y\) for the level-\(r\) Sherali--Adams relaxation.
We will construct an equivalent \(k\)-local pseudo-expectation \(\pE_x\) for \(k=r\).
We define \(\pE_x\) as follows, \[
\pE\nolimits_x f = \pE\nolimits_y (f \circ \psi)\mcom
\] where \(\psi(y)_i=1\) if \(y_{i,1}=0\) and \(\psi(y)_i=-1\) if \(y_{i,1}=1\).
In words, we assign \(1\) to all vertices \(i\) on the same side of the bipartition as vertex \(1\) and \(-1\) to all vertices on the other side.
(Algebraically, \(\psi\) is defined by \(\psi(y)_i=1-2y_{i,1}\).)
This linear functional satisfies \(\pE_x 1=\pE_y 1 = 1\) and \(\pE_x f \ge 0\) for every nonnegative \(k\)-junta \(f\) (because \(f\circ \psi\) is also a nonnegative \(k\)-junta).
It remains to show that \(\pE_x \sum_{ij\in E(G)} (1-x_ix_j)/2=\pE_y \sum_{ij \in E(G)} y_{i,j}\).
By our construction of the functional \(\pE_x\), we have \(\pE_x (1-x_ix_j)/2=\pE_y (y_{i,1}-y_{j,1})^2\) and thus it's enough to establish \(\pE_y (y_{i,1}-y_{j,1})^2- y_{i,j}=0\).
To simplify notation let us assume \(i=2\) and \(j=3\).
Let \(\Ind_{000},\ldots,\Ind_{111}\) be the indicators for the eight possible assignments for the variables \(y_{1,2},y_{1,3},y_{2,3}\).
Since \(1=\Ind_{000}+\cdots+\Ind_{111}\) as functions over \(\{0,1\}^{\binom{n}2}\), it is enough to verify that \(\pE_y \Ind_{abc}\cdot \left ( (y_{1,2}-y_{1,3})^2-y_{1,2} \right )=0\) for all \(a,b,c\in \{0,1\}\).
Note that the identity \((y_{1,2}-y_{1,3})^2=y_{2,3}\) holds if \(y\) is one of the vertices of \({\mathrm{CUT}}_n\).
(In words, vertices \(2\) and \(3\) are on different sides of the bipartition if and only if exactly one of them is on the same side as vertex \(1\).)
We claim that either \(\pE_y \Ind_{abc}=0\) or \((a,b,c)\in {\mathrm{CUT}}_3\cap \{0,1\}^3\).
This claim implies the desired identity, \[
\begin{aligned}
\pE_y (y_{1,2}-y_{1,3})^2- y_{2,3}
& = \sum_{a,b,c\in\{0,1\}} \pE_y \Ind_{abc}\cdot \left ( (y_{1,2}-y_{1,3})^2- y_{2,3} \right )\\
& = \sum_{a,b,c\in\{0,1\}} \pE_y \Ind_{abc}\cdot \left ( (a-b)^2- c \right )\\
& = \sum_{(a,b,c)\in {\mathrm{CUT}}_3 \cap \{0,1\}^3} \pE_y \Ind_{abc}\cdot \left ( (a-b)^2- c \right )\quad \text{(by claim)}\\
& = 0\quad \\
\end{aligned}
\] (The second step uses that \(\Ind_{abc}\cdot \left ( (y_{1,2}-y_{1,3})^2-y_{2,3} \right )=\Ind_{abc}\cdot \left ( (a-b)^2-c \right )\) for all \(y\).)
It remains to prove the claim.
Since \({\mathrm{CUT}}_n \cap \{0,1\}^n={\mathrm{METRIC}}_n \cap \{0,1\}^n\), it is enough to show that for every \(a,b,c\in \{0,1\}\) and every valid linear inequality \(\{\ell\ge 0\}\) for \({\mathrm{METRIC}}_3\) either \(\pE_y \Ind_{abc}=0\) or \(\ell(a,b,c)\ge 0\).
Indeed, since \(\Ind_{abc}\) is a nonnegative \(3\)-junta and \(\{\ell(y_{1,2},y_{1,3},y_{2,3})\}\) is a valid linear inequality for \({\mathrm{METRIC}}_n\), \[
0 \le \pE_{y} \Ind_{abc} \cdot \ell(y_{1,2},y_{1,3},y_{2,3})
=\pE_{y} \Ind_{abc} \cdot \ell(a,b,c)\mcom
\] which means that either \(\pE_y \Ind_{abc}=0\) or \(\ell(a,b,c)\ge 0\).
(The second step uses that \(\Ind_{abc}\cdot \ell(y_{1,2},y_{1,3},y_{2,3})=\Ind_{abc}\cdot\ell(a,b,c)\) for all \(y\).)

\Dnote{TODO: add some (short) discussion about different versions of Sherali--Adams for general MAX CSPs?}

\end{document}

Let $\inst$ be a Sherali--Adams integrality gap instance of \maxcut.

Suppose $\{X_{S}\}_{S \in [n], |S| \leq t }$ is the $t$-round
Sherali--Adams solution on $\inst$.

Construct a new graph $\inst' = \inst_1 \cup \inst_2$ consisting of two disjoint copies of the
instance $\inst$.  Let $N = 2n$ denote the number of vertices of
$\inst'$.
Let us suppose that $\inst_1$ is on vertices $\{1,\ldots, n\}$ and
$\inst_2$
is on $\{n+1, \ldots, 2n\}$.

We will now extend the Sherali--Adams LP solution for $\inst$ to a Sherali--Adams
solution for $\inst'$.  Roughly speaking, we will copy the SA-solution as
is on to $\inst_1$ and negate all its values on $\inst_2$.  In other words, the
Sherali--Adams solution on $\inst'$ is so designed that every pair of
vertices $x_i, x_{n+i}$ always have opposite values.

Formally, for
any subset $S \sse [2n]$ define $S_1 = S \cap \{1,\ldots, n\}$ and
$S_2
= \{ i-n \mid i \in S \setminus S_1 \}$.  Then, we can describe the SA solution
$\{Y_{S}\}_{S \sse [2n], |S| \leq t}$ as follows,
$$   Y_{S} =  (-1)^{|S_2|} X_{S_1 \symdif S_2} $$
Here  $S_1 \symdif S_2$ is the symmetric difference between the two sets.

\begin{definition}
	For a set $S \sse [2n]$ of the form $S = S_1 \cup (S_2 + n)$
	where $S_1, S_2 \sse [n]$, define $\wt(S) = |S_1 \symdif
	S_2|$.
\end{definition}

More generally, we will define $\{ Y_{S} \}_{S \sse [2n]}$ for all
subsets $S \sse [2n]$.
$$ Y_S =
\begin{cases}
	(-1)^{|S_2|}X_{S_1 \symdif S_2} & \text{ if } \wt(S) \leq t \\
         0 & \text{ otherwise}
\end{cases}
$$

\begin{observation} \label{obs:wt}
	For any set $S$, $\wt(S \symdif \{i\}) = \wt(S \symdif
	\{n+i\})$.
\end{observation}
\begin{proof}
	It is easy to see that both values are equal to the hamming weight of
	the $\mathbb{F}_2$ vector $\Ind_{S_1} \oplus \Ind_{S_2} \oplus e_i$ where
	$\Ind_{S_1}, \Ind_{S_2} \in \mathbb{F}_2^n$ are indicators of subsets
	$S_1,S_2$ and $e_i \in \mathbb{F}_2^n$ is the $i^{th}$ standard basis vector.
\end{proof}

Define a function $H : \sbits^N \to \R$ as
\begin{equation}\label{eq:prH}
H(x) = \sum_{S \sse [N]} Y_S \chi_S(x)\,,
\end{equation}
where $\{\chi_S\}$ is the Fourier basis over $\sbits^N$.

\begin{lemma} \label{lem:cancellation}
	For every $S \sse [2n]$ and all $i \in [n]$, $Y_{S\symdif \{i\}} = -Y_{S \symdif
		\{n+i\}}$.
\end{lemma}
\begin{proof}
Let $S = S_1 \cup (S_2 + n)$ for some $S_1, S_2 \sse [n]$.
	By Observation \prettyref{obs:wt}, $\wt(S \symdif \{i\}) = \wt(S \symdif
	\{ n+i \})$.  If $\wt(S \symdif \{i\}) =  \wt(S \symdif
	\{ n+i \}) > t$, then by definition we will have $Y_{S \symdif
		\{i\}} = Y_{S \symdif
		\{n+i\}} = 0$.

	On the other hand, if $\wt(S \symdif \{i\}) =  \wt(S \symdif
	\{ n+i \}) \leq t$ then,
	$$ Y_{S \symdif \{i\}} = (-1)^{|S_2|} X_{ (S_1 \symdif
		\{i\}) \symdif S_2} $$
	and
	$$ Y_{S \symdif \{n+i\}} = (-1)^{|S_2 \symdif \{i\}|} X_{ S_1 \symdif
		(\{i\} \symdif S_2)} $$
	Therefore, also in this case $Y_{S \symdif \{i\}} = -Y_{S
		\symdif \{n+i\}}$.
\end{proof}

\begin{lemma}
	For any polynomial $p(x_1,\ldots, x_N)$, we have
	$\iprod{H,(\sum_{i=1}^N x_i)p} = 0$.
\end{lemma}
\begin{proof}
By linearity of the inner product, it is sufficient to prove the above claim
when $p$ is a monomial.  Recall that,
$$ \iprod{H, \left(\sum_i x_i\right)p} = \E_{x \in \sbits^N}
\left[H(x)\left(\sum_i x_i\right) p(x) \right] $$
Since $x$ takes values in $\sbits^N$, it is sufficient to show the
above claim for the elements of the Fourier basis  $\{ \chi_S \}$.

Fix a monomial $\chi_S(x)$.  We can write the above inner product as,
\begin{align*}
\left\langle \sum_{A \sse [N]} Y_A \chi_A(x), \chi_S(x) \Paren{\sum_i x_i}\right\rangle
& = \E_{x \in \sbits^N} \left[ \chi_S(x) \sum_{A \sse [N], i \in [2n]}
Y_A \chi_{A \symdif \{i\}}(x) \right] \\
&= \E_{x \in \sbits^N} \left[ \chi_S(x) \sum_{B \sse [N]}
\chi_B(x) \left(\sum_{ i \in [2n]}
Y_{B \symdif \{i\}} \right) \right] \\
& = 0 \quad \because{  $\sum_{ i \in [2n]}
Y_{B \symdif \{i\}} = 0$ since $Y_{B \symdif \{i\}} = -Y_{B \symdif \{n+i\}}$}
\end{align*}
\end{proof}

\begin{lemma} \label{lem:nonnegative}
	If $f : \sbits^N \to \R_{\geq 0}$ is a function that depends on a
	subset $J \sse [N]$ of at most $t$ coordinates and possibly
the value $\sum_{i=1}^N x_i$, then
	$$ \iprod{H, f} \geq 0\,.$$
\end{lemma}
\begin{proof}
	Write the function $f$ as a polynomial in $x_J = \{x_i \mid i \in
	J\}$ and $\sum_i x_i$ as follows,
	$$ f = p_0(x_J) + \sum_{i = 1}^{N} p_i(x_J) \left(\sum_i x_i
	\right)^i. $$
	Using \prettyref{lem:cancellation}, we have
	$$ \iprod{H, f} = \iprod{H, p_0}\,.$$
	Since $p_0$ depends on at most $t$ coordinates, by Observation \prettyref{obs:Sherali--Adams} we can write,
	$$ \iprod{H, p_0} = \E_{ x_J \sim \mu_{J} } [p_0(x_J)]\,,$$
    where $\mu_J$ is some distribution on $x_J$.

	Define a distribution $\mu$ on $\sbits^N$ as follows: Sample $x_J$ from
	$\mu_J$ and then sample $x_{\bar{J}}$ uniformly randomly from
	among all assignments that satisfy $\sum_{i=1}^N x_i = 0$.
	This is feasible since $|J| = t < n/2$.  Note that the
	distribution $\mu$ is supported entirely on the set $\{ x \in
		\sbits^N \mid \sum_i x_i
	= 0\}$.
Therefore, we have
	\begin{eqnarray*}
	\iprod{H, p_0} & = &\E_{ x_J \sim \mu_{J} } [p_0(x_J)] \\
&=& \E_{x \sim \mu}
	[p_0(x_J)] \\
	& = &\E_{x \sim \mu} \left[ p_0(x_J) + \sum_{i = 1}^{N} p_i(x_J) \left(\sum_i x_i
		\right)^i \right] \quad \because{$\sum_i x_i = 0$ on the
		support of $\mu$}\\
	& =& \E_{x \sim \mu} [f(x)] \geq 0  \quad\qquad\qquad\qquad\qquad \because{$f$ is
	pointwise non-negative}\qedhere
	\end{eqnarray*}
\end{proof}

    Consider a graph $\inst$ on $n=N/2$ nodes with $\sa_t(\inst) > c$, and let
$\inst'$ be the graph obtained by taking two copies of $\inst$ as discussed before.
    Then $\inst'$ has the property that $\sa_t(\inst')=\sa_t(\inst)$
    and $\opt(\inst')=\opt(\inst)$.
   Let $H$ be the corresponding functional defined in \eqref{eq:prH}.

    Let us consider $\inst'$ as a function on $\sbits^N$ assigning
    cuts to their \maxcut value in $\inst'$.
	Suppose we can express
	express,
	$$ c - \inst' = \sum_{i=1}^M \lambda_i f_i\,,$$
	wherein $\lambda_i \geq 0$.  Taking inner product with the
	functional $H$ on both sides yields
	$$ \iprod{H, c- \inst'} = c - \sa_t(\inst') < 0$$
	while,
	$$ \iprod{H, f_i} \geq 0  \quad \forall i \quad \text{ by
		\prettyref{lem:nonnegative}} \mcom $$
	a contradiction.

\subsection{Truncation}

\begin{definition}
Let $f$ be a distribution on $\sbits^n$ and a set $B \sse \sbits^n$
define a distribution $T_B f$ as,
$$ T_B f(x) = \begin{cases} f(x) + E_{x}[1_B(x) f(x)] & x \notin B \\
	E_{x}[1_{B}(x)f(x)] & x \in B  \end{cases} $$
\end{definition}
From the above definition, it is easy to check that $E[T_B f] = E[f]$.
Notice that for a non-negative function $f$, $T_B f$ is always lower
bounded by $\E_{x}[1_B(x) f(x)]$.  Moreover the functions $f$ and $T_B
f - \E_{x}[1_B(x) f(x)]$ differ only on the set $B$, i.e., the
function $f_B$ given by
$$ f_B =  f - (T_B f - \E[1_{B}(x) f(x)]) $$
is supported only on $B$.

\Pnote{this claim looks ugly}
\begin{claim}
For any function $H : \sbits^n \to \R$ and a subset $B \sse \sbits^n$,
$$ \iprod{1_{\bar{B}}H,f} = \iprod{H, T_B f - \E[1_{B}(x) f(x)]}
$$
\end{claim}
\begin{proof}
	\begin{align*}
		\iprod{1_{\bar{B}}H,f}
		&= \iprod{1_{\bar{B}}H, T_B f - \E[1_{B}(x) f(x)]}
		+ \iprod{1_{\bar{B}}H,f_B} \\
		&= \iprod{1_{\bar{B}}H, T_B f - \E[1_{B}(x) f(x)]}
			\quad \because{ $\supp(f_B) \cap
				\supp(1_{\bar{B}}H) = 0$} \\
		&= \iprod{H, T_B f - \E[1_{B}(x) f(x)]}
			\quad \because{ $\supp(T_B f - \E[1_{B}(x)
			f(x)]) = B$}
	\end{align*}

\end{proof}

Given the family of distributions $\cF = \{f_1,f_2,\ldots, f_M\}$.
Let $B_i$ denote the set of ``bad points'' wherein the value of $f_i$
are exceptionally large.  Formally, define $B_i$ to be
$$ B_i = \{ \text{set of } \lceil 2^{n}/n^t \rceil \text{ points at which } f_i \text{
is largest.} \}$$
Let $B = \cup_{i \in M} B_i$.

\begin{claim}
	For each distribution $f_i \in \cF$, the distribution $T_B f$
	satisfies $T_B f_i(x) \leq n^t + 1$ for all $x \in \sbits^n$.  In particular, $H(T_{B}
	f_i) \geq n - 2t\log{n}$ for all $i \in [M]$.
\end{claim}
\begin{proof}
	By definition of the set $B$, it contains the $2^{n}/n^t$
	largest values of $f_i$.  Since $E[f] = 1$ this implies that
	$f(x) \leq n^t$ for all $x \notin B$.  Further, we have $E[1_B f] \leq \E[f] = 1$
       from which it follows that $T_B f_i(x) \leq n^t + 1$ $\forall
       x \in \sbits^n$.
\end{proof}

Let $\cF' = \{g_1,\ldots, g_M\}$ be a family of distributions over
$\sbits^n$ such that each $g_i$ has high entropy.  Specifically, let
us assume $H(g_i) \geq n -2t \log{n}$ for some $t$.

For a distribution $g$ on $\sbits^n$ and a subset $A \sse [n]$, let
$g_{|A}$ denote the marginal distribution on $A$.  Let $\cF'_{|A}$ denote the family of
distributions $\cF'$ restricted to the set of coordinates in $A$.

\begin{lemma}
For every choice of $m,B,t,d,\beta, n$ and $M$ satisfying
	$$ m B^2 d^2/\beta^2 \leq
	\sqrt{2n} \qquad \text{ and } M < n^t$$
the following holds.
Let $\cF' = \{g_1,\ldots, g_M\}$ be a family of distributions each of
which has entropy $\geq n - B$.  If a set $A \sse [n]$ is
chosen including each element independently with probability $m/n$,
then with constant probability all of the following hold:
\begin{itemize}
	\item $|A| \in [m/2, 2m]$.

	\item  For each $g \in \cF'$, there exists a set of $2t$
		coordinates $J_g \sse A$ such that under
		the distribution $g$ all $d$-wise
		correlations in $A - J_g$ are small.  Quantitatively, we will have
		$$ \E_{S} [ g(x) \chi_S(x)] \leq O(\sqrt{\beta}) \qquad \forall S
		\not\subset J_g, |S| \leq d $$
\end{itemize}
\end{lemma}
\begin{proof}
	First, by Hoeffding's inequality we have $|A| \in [m/2,m]$
	with all but a probability exponentially small in $m$.

	By \prettyref{lem:dist-junta}, for each $g_i$ there exists a
	set $J_i$ of $O(\frac{B^2 d^2}{\beta^2})$ coordinates such
	that all correlations of up to $d$-bits not contained in $J_i$ are low.

	Clearly,$ \E[ |J_i \cap A| ] = \frac{m}{n} |J_i| < 1$.  By
\prettyref{obs:silly-bound}, we can write,
	$$ \Pr{ \left[|J_i \cap A| \geq  2t
	\right]} \leq 2 \left(\frac{m}{n} \cdot |J_i|\right)^{2t} =
	O\left(\left(\frac{m
		B^{2}d^2}{n\beta^2}\right)^{2t}\right) \leq 1/(2n)^t$$

	The lemma follows by taking a union bound over all the $M <
	n^t$ distributions in the family $\cF'$.

\end{proof}

%% file: macros.tex
\usepackage{etex}

\usepackage[l2tabu, orthodox]{nag}
\usepackage{comment}

\usepackage{xspace,enumerate}

\usepackage[dvipsnames]{xcolor}

\usepackage[T1]{fontenc}
\usepackage[full]{textcomp}

\usepackage[american]{babel}

\usepackage{mathtools}

\usepackage{bm}

\ifnum\focs=1
\usepackage{amssymb} 
\fi

\usepackage{amsthm}

\newtheorem{theorem}{Theorem}[section]
\newtheorem*{theorem*}{Theorem}

\newtheorem{proposition}[theorem]{Proposition}
\newtheorem*{proposition*}{Proposition}
\newtheorem{lemma}[theorem]{Lemma}
\newtheorem*{lemma*}{Lemma}
\newtheorem{corollary}[theorem]{Corollary}
\newtheorem*{conjecture*}{Conjecture}

\newtheorem*{fact*}{Fact}

\newtheorem*{hypothesis*}{Hypothesis}

\theoremstyle{definition}
\newtheorem{definition}[theorem]{Definition}

\newtheorem{question}[theorem]{Question}

\theoremstyle{remark}
\newtheorem{claim}[theorem]{Claim}
\newtheorem*{claim*}{Claim}
\newtheorem{remark}[theorem]{Remark}
\newtheorem*{remark*}{Remark}
\newtheorem{observation}[theorem]{Observation}
\newtheorem*{observation*}{Observation}

\ifnum\focs=0 
\usepackage[letterpaper,
top=1.3in,
bottom=1.3in,
left=1.6in,
right=1.6in]{geometry}
\fi

\ifnum\focs=0 
\usepackage[varg]{pxfonts} 
\fi

\ifnum\showkeys=1
\usepackage[color]{showkeys}
\fi

\ifnum\showcolorlinks=1
\ifnum\focs=0
\usepackage[
pagebackref,
letterpaper=true,
colorlinks=true,
urlcolor=blue,
linkcolor=blue,
citecolor=OliveGreen,
]{hyperref}
\fi
\ifnum\focs=1 
\usepackage[
letterpaper=true,
colorlinks=true,
urlcolor=blue,
linkcolor=blue,
citecolor=OliveGreen,
]{hyperref}
\fi
\fi

\ifnum\showcolorlinks=0
\usepackage[
letterpaper=true,
colorlinks=false,
pdfborder={0 0 0}
]{hyperref}
\fi

\usepackage{prettyref}

\newcommand{\savehyperref}[2]{\texorpdfstring{\hyperref[#1]{#2}}{#2}}

\newrefformat{eq}{\savehyperref{#1}{\textup{(\ref*{#1})}}}
\newrefformat{lem}{\savehyperref{#1}{Lemma~\ref*{#1}}}
\newrefformat{def}{\savehyperref{#1}{Definition~\ref*{#1}}}
\newrefformat{thm}{\savehyperref{#1}{Theorem~\ref*{#1}}}
\newrefformat{cor}{\savehyperref{#1}{Corollary~\ref*{#1}}}
\newrefformat{cha}{\savehyperref{#1}{Chapter~\ref*{#1}}}
\newrefformat{sec}{\savehyperref{#1}{Section~\ref*{#1}}}
\newrefformat{app}{\savehyperref{#1}{Appendix~\ref*{#1}}}
\newrefformat{tab}{\savehyperref{#1}{Table~\ref*{#1}}}
\newrefformat{fig}{\savehyperref{#1}{Figure~\ref*{#1}}}
\newrefformat{hyp}{\savehyperref{#1}{Hypothesis~\ref*{#1}}}
\newrefformat{alg}{\savehyperref{#1}{Algorithm~\ref*{#1}}}
\newrefformat{rem}{\savehyperref{#1}{Remark~\ref*{#1}}}
\newrefformat{item}{\savehyperref{#1}{Item~\ref*{#1}}}
\newrefformat{step}{\savehyperref{#1}{step~\ref*{#1}}}
\newrefformat{conj}{\savehyperref{#1}{Conjecture~\ref*{#1}}}
\newrefformat{fact}{\savehyperref{#1}{Fact~\ref*{#1}}}
\newrefformat{prop}{\savehyperref{#1}{Proposition~\ref*{#1}}}
\newrefformat{prob}{\savehyperref{#1}{Problem~\ref*{#1}}}
\newrefformat{claim}{\savehyperref{#1}{Claim~\ref*{#1}}}
\newrefformat{relax}{\savehyperref{#1}{Relaxation~\ref*{#1}}}
\newrefformat{red}{\savehyperref{#1}{Reduction~\ref*{#1}}}
\newrefformat{part}{\savehyperref{#1}{Part~\ref*{#1}}}

\newcommand{\Sref}[1]{\hyperref[#1]{\S\ref*{#1}}}

\usepackage{nicefrac}

\let\nfrac=\nicefrac

\ifnum\usemicrotype=1
\usepackage{microtype}
\fi

\ifnum\showauthornotes=1
\newcommand{\Authornote}[2]{{\sffamily\small\color{red}{[#1: #2]}}}
\newcommand{\Authorcomment}[2]{{\sffamily\small\color{gray}{[#1: #2]}}}
\newcommand{\Authorstartcomment}[1]{\sffamily\small\color{gray}[#1: }

\newcommand{\Authorfnote}[2]{\footnote{\color{red}{#1: #2}}}
\newcommand{\Authorfixme}[1]{\Authornote{#1}{\textbf{??}}}
\newcommand{\Authormarginmark}[1]{\marginpar{\textcolor{red}{\fbox{\Large #1:!}}}}
\else
\newcommand{\Authornote}[2]{}
\newcommand{\Authorcomment}[2]{}
\newcommand{\Authorstartcomment}[1]{}

\newcommand{\Authorfnote}[2]{}
\newcommand{\Authorfixme}[1]{}
\newcommand{\Authormarginmark}[1]{}
\fi

\newcommand{\Dnote}{\Authornote{D}}

\newcommand{\Jnote}{\Authornote{J}}

\newcommand{\Pnote}{\Authornote{P}}

\ifnum\showfixme=0

\fi

\usepackage{boxedminipage}

\newcommand{\Paren}[1]{\left(#1\right)}

\newcommand{\abs}[1]{\lvert#1\rvert}
\newcommand{\Abs}[1]{\left\lvert#1\right\rvert}

\newcommand{\card}[1]{\lvert#1\rvert}

\newcommand{\set}[1]{\{#1\}}

\newcommand{\Bigset}[1]{\Big\{#1\Big\}}
\newcommand{\norm}[1]{\lVert#1\rVert}

\newcommand{\normo}[1]{\norm{#1}_1}

\newcommand{\iprod}[1]{\langle#1\rangle}

\newcommand{\kldiv}[2]{D(#1\,\Vert\,#2)}

\newcommand{\Esymb}{\mathbb{E}}
\newcommand{\Psymb}{\mathbb{P}}

\DeclareMathOperator*{\E}{\Esymb}

\DeclareMathOperator*{\ProbOp}{\Psymb}

\renewcommand{\Pr}{\ProbOp}

\ifnum\focs=0 
\usepackage{dsfont}
\usepackage{mathrsfs}
\fi

\newcommand{\textparen}[1]{\text{(#1)}}

\ifx\because\undefined
\newcommand{\because}[1]{\textparen{because #1}}
\else
\renewcommand{\because}[1]{\textparen{because #1}}
\fi

\newcommand{\bits}{\{0,1\}}
\newcommand{\sbits}{\{-1,1\}}

\newcommand{\sm}{\setminus}

\newcommand{\defeq}{\stackrel{\mathrm{def}}=}

\newcommand{\seteq}{\mathrel{\mathop:}=}

\newcommand{\from}{\colon}

\newcommand{\Bigmid}{~\Big|~}

\newcommand{\mper}{\,.}
\newcommand{\mcom}{\,,}

\newcommand\bdot\bullet

\newcommand{\Ind}{\mathbb I}

\DeclareMathOperator{\opt}{opt}

\DeclareMathOperator{\cone}{cone}

\DeclareMathOperator{\supp}{supp}
\DeclareMathOperator{\Sym}{\bm{Sym}}
\DeclareMathOperator{\Stab}{Stab}

\DeclareMathOperator{\rank}{rank}

\newcommand{\N}{\mathbb N}
\newcommand{\R}{\mathbb R}

\newcommand{\Rnn}{\R_+}

\newcommand{\problemmacro}[1]{\texorpdfstring{\textsc{#1}}{#1}\xspace}

\newcommand{\maxcut}{\problemmacro{Max Cut}}

\newcommand{\vertexcover}{\problemmacro{Vertex Cover}}

\newcommand{\maxtwosat}{\problemmacro{Max $2$-Sat}}
\newcommand{\maxthreesat}{\problemmacro{Max $3$-Sat}}

\newcommand{\maxclique}{\problemmacro{Max Clique}}

\newcommand{\cF}{\mathcal F}
\newcommand{\cG}{\mathcal G}

\newcommand{\cL}{\mathcal L}

\newcommand{\cX}{\mathcal X}

\renewcommand{\leq}{\leqslant}
\renewcommand{\le}{\leqslant}
\renewcommand{\geq}{\geqslant}
\renewcommand{\ge}{\geqslant}

\ifnum\showdraftbox=1

\else

\fi

\let\epsilon=\varepsilon

\numberwithin{equation}{section}

\newcommand{\MYstore}[2]{%
  \global\expandafter \def \csname MYMEMORY #1 \endcsname{#2}%
}

\newcommand{\MYload}[1]{%
  \csname MYMEMORY #1 \endcsname%
}

\newcommand{\MYnewlabel}[1]{%
  \newcommand\MYcurrentlabel{#1}%
  \MYoldlabel{#1}%
}

\newcommand{\MYdummylabel}[1]{}

\newcommand{\torestate}[1]{%
  \let\MYoldlabel\label%
  \let\label\MYnewlabel%
  #1%
  \MYstore{\MYcurrentlabel}{#1}%
  \let\label\MYoldlabel%
}

\newcommand{\restatetheorem}[1]{%
  \let\MYoldlabel\label
  \let\label\MYdummylabel
  \begin{theorem*}[Restatement of \prettyref{#1}]
    \MYload{#1}
  \end{theorem*}
  \let\label\MYoldlabel
}

\newcommand{\restatelemma}[1]{%
  \let\MYoldlabel\label
  \let\label\MYdummylabel
  \begin{lemma*}[Restatement of \prettyref{#1}]
    \MYload{#1}
  \end{lemma*}
  \let\label\MYoldlabel
}

\newcommand{\restateprop}[1]{%
  \let\MYoldlabel\label
  \let\label\MYdummylabel
  \begin{proposition*}[Restatement of \prettyref{#1}]
    \MYload{#1}
  \end{proposition*}
  \let\label\MYoldlabel
}

\newcommand{\restatefact}[1]{%
  \let\MYoldlabel\label
  \let\label\MYdummylabel
  \begin{fact*}[Restatement of \prettyref{#1}]
    \MYload{#1}
  \end{fact*}
  \let\label\MYoldlabel
}

\newcommand{\restate}[1]{%
  \let\MYoldlabel\label
  \let\label\MYdummylabel
  \MYload{#1}
  \let\label\MYoldlabel
}

\newcommand{\addreferencesection}{
  \phantomsection
  \addcontentsline{toc}{section}{References}
}

\newcommand{\sse}{\subseteq}

\newcommand{\e}{\epsilon}
\newcommand{\eps}{\epsilon}

\let\origparagraph\paragraph
\renewcommand{\paragraph}[1]{\origparagraph{#1.}}

\allowdisplaybreaks

